\documentclass{elsarticle}

\usepackage{titlesec}

\titleformat*{\section}{\Large\bfseries}
\titleformat*{\subsection}{\large\bfseries}

\usepackage{t1enc}
\usepackage{graphicx}
\usepackage{natbib}
\usepackage{longtable}
\usepackage[table,usenames,dvipsnames]{xcolor}
\usepackage{caption}
\usepackage{multirow}
\usepackage{hyperref}
\usepackage{rotating}
\usepackage{array}
\usepackage{hhline}
\usepackage{dsfont}

\journal{Economic Theory}

\biboptions{authoryear}

\newcolumntype{?}[1]{!{\vrule width #1}}

\newcommand{\vabs}{\widehat{v}}
\newcommand{\vrel}{v}

\newcommand{\wabs}{\widehat{W}}
\newcommand{\wrel}{W}

\newcommand{\wmabs}[1]{absolute-$#1$-maximizer}
\newcommand{\wmrel}[1]{relative-$#1$-maximizer}

\newcommand{\en}{}
\definecolor{myBlue}{RGB}{120,200,210}
\definecolor{myGreen}{RGB}{0,180,0}
\definecolor{myRed}{RGB}{220,0,0}
\definecolor{myOrchid}{RGB}{190,125,219}

\usepackage{amsmath}
\usepackage{amsthm}
\usepackage{amssymb}

\usepackage{thmtools,thm-restate}

\theoremstyle{plain}
\newtheorem{theorem}{Theorem}[section]
\newtheorem*{theorem*}{Theorem}
\newtheorem{lemma}[theorem]{Lemma}
\newtheorem*{lemma*}{Lemma}
\newtheorem{corollary}[theorem]{Corollary}

\theoremstyle{definition}
\newtheorem{example}[theorem]{Example}
\newtheorem{definition}[theorem]{Definition}
\newtheorem{remark}[theorem]{Remark}
\newtheorem*{remark*}{Remark}
\newtheorem*{assumption*}{Assumption}

\def\c#1{\multicolumn{1}{|c|}{#1}}  

\begin{document}

\begin{frontmatter}
\title{Monotonicity and Competitive Equilibrium in Cake-cutting}

\author[biu]{Erel Segal-Halevi}

\author[gt,cor]{Bal\'azs Sziklai}

\address[biu]{Bar-Ilan University, Ramat-Gan 5290002, Israel. erelsgl@gmail.com}

\address[gt]{Game Theory Research Group, Centre for Economic and Regional Studies, Hungarian Academy of Sciences, Email: sziklai.balazs@krtk.mta.hu}

\address[cor]{Corvinus University of Budapest, Department of Operations Research and Actuarial Sciences.}

\begin{abstract}
We study the monotonicity properties of solutions in the classic problem of fair cake-cutting --- dividing a heterogeneous resource among agents with different preferences. Resource- and population-monotonicity relate to scenarios where the cake, or the number of participants who divide the cake, changes. It is required that the utility of all participants change in the same direction: either all of them are better-off (if there is more to share or fewer to share among) or all are worse-off (if there is less to share or more to share among).

We formally introduce these concepts to the cake-cutting problem and examine whether they are satisfied by various common division rules.
We prove that the Nash-optimal rule, which maximizes the product of utilities, is
resource-monotonic and population-monotonic, in addition to being Pareto-optimal, envy-free and satisfying a strong competitive-equilibrium condition.
Moreover, we prove that it is the only rule among a natural family of welfare-maximizing rules that is both proportional and resource-monotonic.
\end{abstract}

\end{frontmatter}

\section{Introduction}
The interest in monotonicity axioms was motivated by paradoxes such as the throw-away paradox \citep{Aumann1974Note}: in some cases an agent can improve his final utility by discarding some goods from the initial endowment.
Such apparent paradoxes occur in real life too. Farmers may want to burn crops in order to increase their market-price.
In traffic networks, it may be possible to decrease the time spent on traffic-jams by removing a bridge (this is the so-called Breass Paradox, \citet{Braess1968}).
In all these cases, some agents are worse-off when the social endowment is larger. Such agents have an incentive to destroy some of the endowment, or at least prevent growth, in order to improve their well-being.

These issues motivated the axiom of  \emph{resource-monotonicity (RM)}.%
\footnote{
Resource-monotonicity was introduced by \citet{Moulin1988}.
It is also common in cooperative game theory,
where it is called \emph{aggregate monotonicity} \citep{Peleg2007}.
}
It requires that when new resources are added, and the same division rule is used consistently, the utility of all agents should weakly increase.
A related axiom is \emph{population-monotonicity (PM)}. It is concerned with changes in the number of participants.
It requires that 
when a new agent joins the process, all existing participants have to make sacrifices in order to support the new agent, thus their utility weakly decreases. Both properties have inverse versions ---
when the cake becomes smaller every utility should weakly decrease (RM), when someone leaves the division process every utility should weakly increase (PM).

The present paper studies these two monotonicity requirements in the framework of the classic \emph{fair cake-cutting} problem \citep{Steinhaus1948}, where a single heterogeneous resource --- such as land or time --- has to be divided fairly. Fair cake-cutting protocols can be applied in inheritance cases and divorce settlements. They can be also used to divide broadcast time of advertisements and priority access time for customers of an Internet service provider \citep{Caragiannis2011}.

The notion of fairness in cake-cutting is commonly restricted to two properties: \emph{proportionality} means that each of the $n$ agents should receive a value of at least $1/n$ of the total cake value; \emph{envy-freeness} means that each agent weakly prefers his share over the share of any other agent.
Monotonicity axioms have not been adapted so far for the cake-cutting literature. Indeed, they are violated by all classic fair-cake-cutting procedures that we checked. For example, the classic cut-and-choose protocol is proportional but not resource-monotonic\footnote{In an accompanying technical report \citep{ourTechReport} we provide similar examples showing that other classic cake-cutting procedures, like Banach-Knaster, Dubins-Spanier and many others, violate both RM and PM.} (Section \ref{sec:examples}).
It is easy to find monotonic rules that are not proportional (e.g.\ the rule that gives the entire cake to a pre-specified agent). Our goal in this paper is to find division rules that satisfy all fairness axioms simultaneously.

Initially (\textbf{Section \ref{sec:maximizers}})
we focus on two natural families of welfare-maximizing rules --- rules that maximize the sum of an increasing function of the absolute or relative values of the agents.  These families include, as special cases, the \emph{utilitarian} rule (maximizing the sum of utilities) and the \emph{Nash-optimal} rule (maximizing the sum of log of utilities).
We prove necessary and sufficient conditions for such rules to be essentially-single-valued (ESV) --- recommend a unique utility-profile --- as well as for monotonicity and proportionality. Based on these conditions, we prove that the Nash-optimal rule is the only rule in these families that is essentially-single-valued, resource-monotonic and proportional.
Moreover, this rule is also population-monotonic and envy-free.
This solves an open question posed by \citet{Berliant1992Fair}\footnote{"...there are a number of important issues that should be tackled next pertaining, in particular, to the existence of selections from the no-envy solution satisfying additional properties, Examples are monotonicity with respect to the amount to be divided (all agents should benefit from such an increase), and with respect to changes in the number of claimants (all agents initially present should lose in such circumstances)." \citep{Berliant1992Fair}}.

Then (\textbf{Section \ref{sec:ceei-nash}})
we focus on another family of rules,
related to the famous rule of \emph{competitive equilibrium from equal incomes (CEEI)}.
This rule is known to be envy-free and Pareto-optimal in other domains, and \citet{Weller1985Fair} proved that it exists in cake-cutting too. However, we prove that a CEEI by Weller's definition (which we call WCEEI) is not Pareto-optimal and not ESV. A stronger variant of this definition, which we call PCEEI, is Pareto-optimal, but still not ESV. We present an even stronger variant, which we call SCEEI, and prove that is Pareto-optimal and ESV, as well as RM, PM, envy-free and proportional. Moreover, we prove that SCEEI is in fact identical to the Nash-optimal rule, and note that this equivalence does not hold for its weaker variants.

Finally (\textbf{Section \ref{sec:leximin}}), we study a third family of rules, based on the \emph{leximin} principle (maximizing the smallest value, then the next-smallest value, etc.)
We study two rules in this family: one based on absolute values and the other on relative values. We show that the absolute-leximin rule is RM but not proportional, while the relative-leximin rule is proportional but not RM.

Our work shows that the Nash-optimal rule is ideal for fair cake cutting.
Other works, published independently and contemporaneously to our work, prove other desirable properties of the Nash-optimal-CEEI rule for indivisible item assignment \citep{Caragiannis2016Unreasonable}, public decision making \citep{Conitzer2017} and  homogeneous resource allocation \citep{Bogomolnaia2016Competitive}. The combined evidence of the these works shows that the Nash-optimal rule may be the most fair allocation rule in various settings.

\section{Related Work}\label{sec:literature}
The cake-cutting problem originates from the work of the Polish mathematician Hugo Steinhaus and his students Banach and Knaster \citep{Steinhaus1948}. Their primary concern was how to divide the cake in a fair way. Since then, game theorists analyzed the strategic issues related to cake-cutting, while computer scientists were focusing mainly on how to
compute
solutions efficiently. See \citet{Branzei2015Computational,Procaccia2015Cake} for recent reviews.

Monotonicity issues have been extensively studied with respect to cooperative game theory \citep{Calleja2012}, political representation \citep{Balinski1982}, computer resource allocation \citep{Ghodsi2011Dominant}, single-peaked preferences \citep{Sonmez1994} and other fair division problems. Extensive reviews of monotonicity axioms can be found in chapters 3, 6 and 7 of \citep{Moulin2004Fair} and in chapter 7 of \cite{Thomson_2011}. To the best of our knowledge, the present paper is the first that studies these properties in a cake-cutting setting.

Experimental studies show that people value certain fairness criteria more than others. \cite{Herreiner2009} demonstrated that people are willing to sacrifice Pareto-efficiency in order to reach an envy-free allocation. To our knowledge no study was ever conducted to unfold the relationship between monotonicity and efficiency or proportionality. However some indirect evidence points toward that monotonicity of the solution is in some cases as important as proportionality. The so called \emph{apportionment problem}, where electoral seats have to be distributed among administrative regions provides the most notorious examples. The seat distribution of the US House of Representatives generated many monotonicity related anomalies in the last two centuries. The famous Alabama-paradox, as well as the later discovered population and new state paradoxes pressed the legislators to adopt newer and newer apportionment rules. The currently used seat distribution method is free from such anomalies, however it does not satisfy the so called Hare-quota, a basic guarantee of proportionality \citep{Koczy2017}. We view this as an evidence that monotonicity is as important fairness axiom as the classic axioms of proportionality and envy-freeness.

\cite{Thomson1997Replacement} defines the \emph{replacement principle}, which requires that, whenever any change happens in the environment, the welfare of all agents not responsible for the change should be affected in the same direction --- they should all be made at least as well off as they were initially or they should all be made at most as well off. This is the most general way of expressing the idea of \emph{solidarity} among agents. The PM and RM axioms are special cases of this principle.

The consistency axiom \citep{Young1987,Thomson_2012} resembles population-monotonicity since in both axioms the set of agents changes. However it is fundamentally different as it assumes that leaving agents take their fair shares with them.

\cite{Arzi2012Cake,Arzi2016Toss} study the "dumping paradox" in cake-cutting. They show that, in some cakes, discarding a part of the cake improves the total social welfare of any envy-free division. This implies that enlarging the cake might decrease the total social welfare. This is related to resource-monotonicity; the difference is that in our case we are interested in the welfare of the individual agents and not in the total social welfare.

\cite{Chambers_2005} studies a related cake-cutting axiom called "division independence": if the cake is divided into sub-plots and each sub-plot is divided according to a rule, then the outcome should be identical to dividing the original cake using the same rule. He proves that the only rule which satisfies Pareto-optimality and division independence is the utilitarian-optimal rule --- the rule which maximizes the sum of the agents' utilities. Unfortunately, this rule does not satisfy the fairness axioms of proportionality and envy-freeness.

\cite{Walsh2011Online} studies the problem of "online cake-cutting", in which agents arrive and depart during the process of dividing the cake. He shows how to adapt classic procedures like cut-and-choose and the Dubins-Spanier in order to satisfy online variants of the fairness axioms. Monotonicity properties are not studied, although the problem is similar in spirit to the concept of population-monotonicity. \cite{Kash2014} also study a dynamic resource allocation setting, but they deal with multiple, homogeneous, divisible resources. The authors assume that participants are added sequentially, but resources allocated to existing participants cannot be taken back, which can be viewed as a stronger form of resource monotonicity.
\citet{segal2016re} studies a related problem of how to re-divide a cake when new agents join the scene, while balancing fairness for the new agents with the ownership rights of the old agents.

\section{Model}  \label{sec:model}
\subsection{Cake-cutting}
A cake-cutting problem is a triple $\Gamma(N,C,(\vabs_i)_{i \in N})$ where:

\begin{itemize}
  \item $N=\{1,2,\dots,n\}$ denotes the set of agents who participate in the cake-cutting process. In examples with a small number of agents, we often refer to them by names (Alice, Bob, Carl...).
  \item $C$ is the cake. We assume that $C$ is an interval, $C=[0,c]$ for some real number $c$.
  \item $\vabs_i$ is the value measure of agent $i$. It is a non-negative real-valued function defined on the Borel subsets of $[0,\infty)$. We assume that the value of every finite interval is finite.
  As the term ``measure'' implies, $\vabs_i$ is countably additive: the value measures of a countable union of disjoint subsets is the sum of the values of the subsets.

\end{itemize}


We call a Borel subset of $C$ a \emph{slice}.
A slice with a positive value for at least one agent is called a \emph{positive slice}.
A slice allotted to an agent is called a \emph{piece}.

We assume that the value measures are nonatomic. This means that any point on the interval is worth 0 for all agents\footnote{
It is sometimes assumed that the value-measures are absolutely-continuous with respect to Lebesgue measure. This means that any slice with zero length has zero value for everyone. This is equivalent to the assumption that each value-measure is the integral of some ``value-density'' function, describing the value per unit of length. The absolute-continuity assumption is strictly stronger than our non-atomicity assumption; see \citet{hill2010cutting} and \citet{schilling2016continuity}.
}. We also assume that each agent values the entire cake as positive.
All these assumptions are standard in the cake-cutting literature.

Our model diverges from the standard cake-cutting setup in that we do not require the value measures to be normalized. That is, the value of the entire cake is not necessarily the same for all agents. This is important because we examine scenarios where the cake changes, so the cake value might become larger or smaller. Hence, we differentiate between \emph{absolute} and \emph{relative} value measures:
\begin{itemize}
\item The \textbf{absolute} value measure of the entire cake, $\vabs_i(C)$, can be any positive value and it can be different for different agents.
\item The \textbf{relative} value of the entire cake is 1 for all agents.
Relative value measures are denoted by $\vrel_i$ and defined by:
$\vrel_i(X):={\vabs_i(X) / \vabs_i(C)}$.
\end{itemize}

It is also common to assume that value measures are private information of the agents. This question leads us to whether agents are honest about their preferences. While cake-cutting problems can be studied from a strategic angle \citep{Branzei2016Algorithmic}, here we will not analyze the strategic behavior of the agents and assume that their valuations are known.

A \emph{division} is a partition of the cake into $n$ pairwise-disjoint pieces, $X=(X_1,\dots, X_n)$, such that $X_1\cup\cdots\cup X_n=C$.

A \emph{division rule} is a correspondence that takes a cake-cutting problem as input and returns a division or a set of divisions.


A division rule $R$ is called \emph{essentially single-valued (ESV)} if $X,Y \in R(\Gamma)$ implies that for all $i \in N$, $\vabs_i(X_i)=\vabs_i(Y_i)$. That is, even if $R$ returns a set of divisions, all agents are indifferent between these divisions.

The classic requirements of fair cake-cutting are the following. A division $X$ is called:

\begin{itemize}
  \item \emph{Pareto-optimal} (PO) if there is no other division which is weakly better for all agents and strictly better for at least one agent.
  \item \emph{Proportional} (PROP) if each agent gets at least $1/n$ fraction of the cake according to his own evaluation, i.e.\ for all $i \in N$, $\vrel_i(X_i)\geq 1/n$. Note that the  definition uses relative values.
  \item \emph{Envy-free} (EF) if each agent gets a piece which is weakly better, for that agent, than all the other agents' pieces: for all $i,j \in N$, $\vrel_i(X_i)\geq \vrel_i(X_j)$. Note that here it is irrelevant whether absolute or relative values are used. Note that, since the entire cake is divided, EF implies PROP.
\end{itemize}

A division rule is called \emph{Pareto-optimal} (PO) if it returns only PO divisions. The same applies to proportionality and envy-freeness.

\subsection{Monotonicity}
We now define the two monotonicity properties. In the introduction we defined them informally for the special case in which the division rule returns a single division. Our formal definition is more general and applicable to rules that may return a set of divisions.

The first two definitions relate to \emph{resource-monotonicity} (RM).

\begin{definition}
Let $N$ be a fixed set of agents, $C=[0,c]$, $C'=[0,c']$ two cakes where $c<c'$, and $(\vabs_i)_{i \in N}$ value measures on $[0,\infty)$. Then the cake-cutting problem $\Gamma'=(N, C', (\vabs_i)_{i \in N})$ is called a \emph{cake-enlargement} of $\Gamma=(N, C, (\vabs_i)_{i \in N})$. The set of cake-enlargements of $\Gamma$ is denoted $\textsc{CakeEnlargements}(\Gamma)$.
\end{definition}

We enlarge the cake from the right-hand side for practical reasons, but this fact does not have any significance from a theoretical point of view (the theorems are valid no matter where the enlargement is placed).


\begin{definition}
A division rule $R$ is called:

(a) \emph{Upwards RM} --- if for every $\Gamma$, every $\Gamma'\in \textsc{CakeEnlargements}(\Gamma)$ and every division $X\in R(\Gamma)$ there \emph{exists} a division $Y \in R(\Gamma')$ such that $\vabs_i(Y_i) \geq \vabs_i(X_i)$ for all $i \in N$ (all agents are weakly better-off when the cake is larger).

(b) \emph{Downwards RM} ---
if for every $\Gamma$, every $\Gamma'\in \textsc{CakeEnlargements}(\Gamma)$ and every  $Y\in R(\Gamma')$ there \emph{exists} a division $X \in R(\Gamma)$ such that $\vabs_i(X_i) \leq \vabs_i(Y_i)$ for all $i \in N$ (all agents are weakly worse-off when the cake is smaller).

(c) A division rule is RM if it is both upwards RM and downwards RM.
\end{definition}

The next two definitions relate to \emph{population-monotonicity} (PM).

\begin{definition}
Let $C$ be a fixed cake, $N$ and $N'$ two sets of agents such that $N\supset N'$ and $(\vabs_i)_{i \in N}$ their value measures.
Then the cake-cutting problem $\Gamma'=(N', C, (\vabs_i)_{i \in N'})$ is called a \emph{population-reduction} of $\Gamma=(N, C, (\vabs_i)_{i \in N})$. The set of population-reductions of $\Gamma$ is denoted $\textsc{PopReductions}(\Gamma)$.
\end{definition}

\begin{definition} A division rule $R$ is called:

(a) \emph{Upwards PM}, if for every $\Gamma$,
every $\Gamma'\in\textsc{PopReductions}(\Gamma)$ and every division $Y \in R(\Gamma')$, there exists a division $X \in R(\Gamma)$ such that $\vabs_i(X_i) \leq \vabs_i(Y_i)$ for all $i \in N'$ (all agents who participate in both divisions are weakly worse-off when a new agent joins the process).

(b) \emph{Downwards PM}, if for every $\Gamma$,
every $\Gamma'\in\textsc{PopReductions}(\Gamma)$ and every division $X \in R(\Gamma)$, there exists a division $Y \in R(\Gamma')$ such that $\vabs_i(Y_i) \geq \vabs_i(X_i)$ for all $i \in N'$ (all agents who participate in both divisions are weakly better-off if someone leaves).

(c) A division rule is \emph{population-monotonic} (PM), if it is both upwards and downwards population-monotonic.
\end{definition}

\begin{remark}
As usual in the literature, the monotonicity axioms care only about absolute values. It is not considered a violation of RM if the relative value of an agent decreases when the cake grows.
\end{remark}

\begin{remark} For essentially-single-valued solutions, downwards resource (or population) monotonicity implies upwards resource (or population) monotonicity and vice versa. Set-valued solutions, however, may satisfy only one direction of these axioms.
\end{remark}

\begin{remark}
The monotonicity axioms in \cite{Thomson_2011} require that \emph{all} divisions in $R(\Gamma)$ should be weakly better/worse than all divisions in $R(\Gamma')$.
In contrast, our definition,
which originates from cooperative game theory \citep{Peleg2007},
only requires that there \emph{exists} such a division.
Clearly, our monotonicity condition is weaker and it is implied by Thomson's monotonicity.
The rationale behind the weaker definition
is that even if a set-valued solution is used, only a single allocation will be implemented. Hence, the divider can be faithful to the monotonicity principles even if the rule suggests many non-monotonic allocations as well.
However,
the ``protagonist'' of the present paper
---
the Nash-optimal rule
---
is essentially-single-valued (Corollary \ref{cor:nash}); for such rules, the two definitions coincide.

\end{remark}

\subsection{Examples}
\label{sec:examples}

In our first example we show that Cut and Choose, the most widely applied division method, is not resource monotonic. Suppose Alice and Bob want to divide a one-dimensional land-plot, e.g. a river-bank, that is modeled by the cake $C_1=[0,5]$.
Alice's valuation of the land increases the more eastwards we go. Her valuation is given by (upper curves in Figure \ref{nhc}):
\begin{align*}
\vabs_A([0,x])&=x^2+7x, &\frac{d\hat{v}_A([0,x])}{d x}&=2x+7.
\end{align*}
Bob likes the lands adjacent to $x={7 / 2}$. His valuation is given by (lower curves in Figure \ref{nhc}):
\begin{align*}
\vabs_B([0,x])&=3\arctan\left(4x-14\right)+3\arctan(14), &\frac{d\hat{v}_B([0,x])}{d x}&=\frac{12}{16\left(x-\frac{7}{2}\right)^2+1}.
\end{align*}
They play Cut and Choose: Alice cuts the cake into two pieces, Bob chooses the piece most valuable for him and Alice keeps the other one. Figure~\ref{nhc} illustrates the process. On the right hand side, the function $\vabs([0,x])$ represents the cumulative utility of the agents, i.e.\ the value of the piece that lies left to the point $x\in \mathbb{R}$. On the left hand side, we plotted the derivative of this function, which can be interpreted as a ``value density''.

Alice cuts $C_1$ at 3, as the pieces $[0,3]$ and $[3,5]$ worth the same for her. It is clear from the figure that the most valuable part of the cake for Bob lies right to the cut point, from which Bob obtains a utility of $\int_{3}^{5} \frac{12}{16\left(x-\frac{7}{2}\right)^2+1}d\mu\approx7.53$.

Now, new lands become available to the east of the existing lands, so the cake is now $C_2=[0,6]$. Alice cuts $C_2$ around 3.65 and Bob is forced to choose between the pieces $[0,3.65]$ and $[3.65,6]$. Since $\int_{0}^{3.65} \frac{12}{16\left(x-\frac{7}{2}\right)^2+1}d\mu\approx6.11$ and $\int_{3.65}^{6} \frac{12}{16\left(x-\frac{7}{2}\right)^2+1}d\mu\approx2.79$ Bob loses utility no matter which one he chooses.

\begin{figure}
  \centering
  \includegraphics[width=12.5cm]{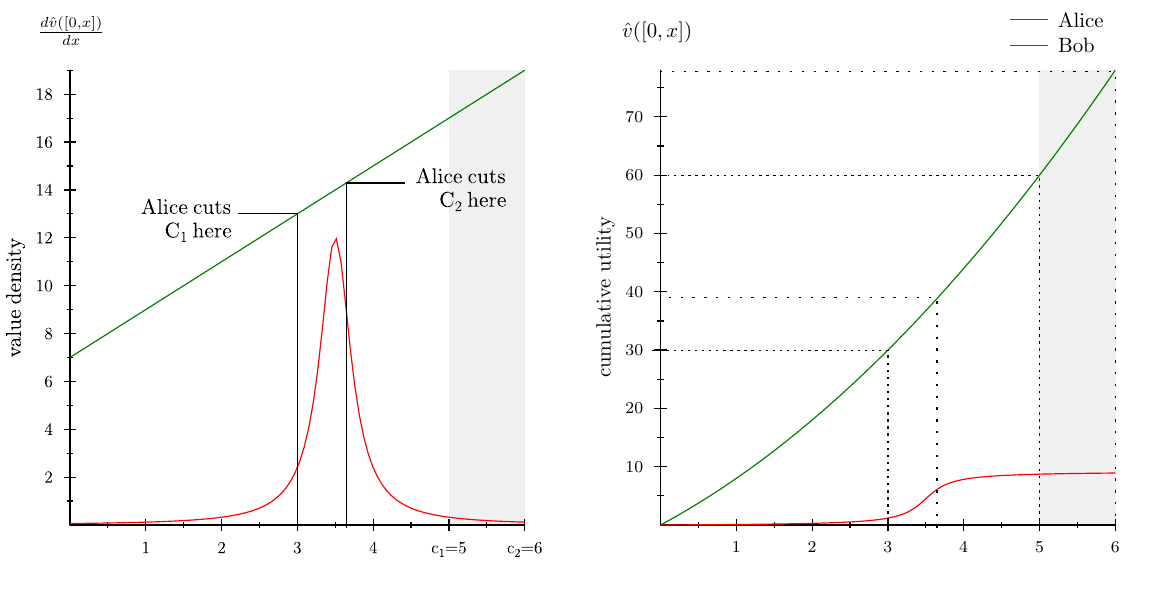}\\
  \caption{Cut and Choose is not RM.}\label{nhc}
\end{figure}


This simple counterexample implies that the Banach-Knaster and the Fink methods are not resource-monotonic either, since they both produce the same division on the above cake as Cut and Choose\footnote{A more recent protocol, the Recursive Cut and Choose, proposed by \cite{Tasnadi2003}, violates resource-monotonicity for the same reason.}. It is an easy exercise to show that Bob loses utility with the Dubins-Spanier or the Even-Paz methods as well.

\begin{remark}
\citet{ourArxivPaperConnected} show that the Banach-Knaster, Cut and Choose, Dubins-Spanier, Even-Paz, Fink, and Selfridge-Conway methods are neither RM nor PM. Some of these methods can be made monotonic by using a special ordering: the Fink method is upwards-PM if the new agent is the last who chooses slices; the Cut and Choose protocol is RM if the agents are ordered by their cut marks and the one who has the rightmost cut mark cuts the cake
\footnote{this variant is called the \emph{Rightmost Mark} rule \citep{ourArxivPaperConnected}.}
.
\end{remark}

The rest of the paper will feature \emph{piecewise homogeneous} cakes. A piecewise-homogeneous cake is a finite union of disjoint intervals, such that inside each interval $I_j$, every agent $i$ values each subset $Z\subseteq I_j$ as $(a_{i,I_j}\cdot \operatorname{length}(Z))$, where the $a_{i,I_j}$ values are constants. In such cases, the cumulative utility of the agents, $\vabs([0,x])$ is a piecewise-linear function (see Figure \ref{plvf4}). Although piecewise-homogeneous cakes lack a certain cake cutting flavour, their simplicity makes them ideal for illustration purposes. We stress that our theorems hold for arbitrary cakes --- not only for piecewise-homogeneous ones.

\begin{figure}
  \centering
  \includegraphics[width=12.5cm]{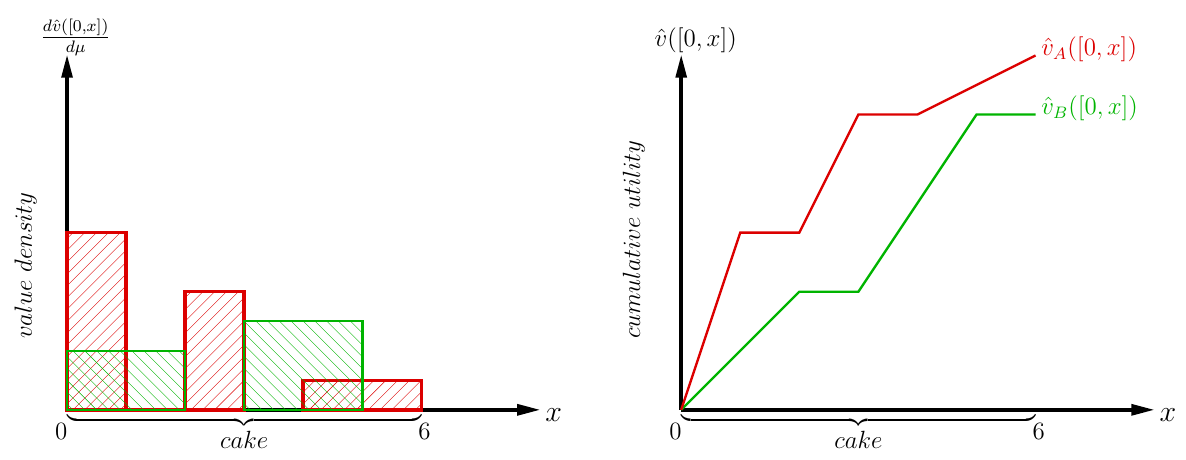}\\
  \caption{Piecewise homogeneous cake with two players. Note that the value measures are not normalized, hence $\vabs_A([0,6])\neq \vabs_B([0,6])$.}\label{plvf4}
\end{figure}

Piecewise homogeneous cakes can be represented by a table containing the values of the agents on the different intervals. For example the cake in Figure \ref{plvf4} can be presented as follows (note that each interval has a length of 1):

\begin{center}
  \centering
  \begin{tabular}{|c|c|c|c|c|c|c|}
    \hline
    $\vabs_A$ & 6 & 0 & 4  & 0 & 1  & 1 \\
    \hline
    $\vabs_B$ & 2 & 2 & 0  & 3 & 3 & 0 \\
    \hline
  \end{tabular}
\end{center}

Throughout the paper the $\blacktriangledown$ sign over a column indicates the enlargement.

\section{Welfare-Maximizing Rules} \label{sec:maximizers}
In this section we search for division rules satisfying the desired properties of PROP, EF, RM, PM and PO.
Since Pareto-optimality is a basic requirement,
we confine our search to two families of division-rules that are PO by design --- the families of \emph{welfare maximizing} rules. A welfare-maximizing rule is parametrized by a strictly increasing function $w: \mathbb{R}^+\to \mathbb{R}\cup\{-\infty,+\infty\}$. We limit our attention to welfare functions that are twice continuously differentiable. We also assume that $w$ is independent of $n$; such consistency requirement is natural when investigating rules defined over varying population or resource sizes.

Given a function $w$, it is possible to define two different social welfare maximizers:

\begin{itemize}
\item The \emph{\wmrel{w}} rule selects all allocations that maximize the \emph{relative-welfare} function $\wrel_w(X) := \sum_{i=1}^n w(\vrel_i(X_i))$;
\item The \emph{\wmabs{w}} rule selects allocations that maximize the \emph{absolute-welfare} function $\wabs_w(X) := \sum_{i=1}^n w(\vabs_i(X_i))$.
\end{itemize}
Such social-welfare functions satisfy several reasonable axioms
\cite[pages 66-69]{Moulin2004Fair}.
The following lemma shows that these rules are well-defined.
\begin{lemma}
\label{lem:existence}
For every continuous function $w$, there exists an allocation that maximizes $\wrel_w$ and an allocation that maximizes  $\wabs_w$.
\end{lemma}
\begin{proof}
For every allocation $X$, let $M(X)$ be its \emph{relative-value-matrix} --- a matrix $M$ where  $\forall i,j: M_{i,j} = \vrel_i(X_j)$.
Let $\mathbb{M}$ the set of all matrices that correspond to allocations of the cake:
\begin{align*}
\mathbb{M} :=
\{
	M(X)
	| X \text{ is a cake-allocation}
\}
\end{align*}
Theorem 1 of \citet{Dubins_1961}, which is a special case of a theorem of \citet{dvoretzky1951relations}, implies that, if all value-measures are non-atomic (as we assume throughout the paper), then
$\mathbb{M}$ is a compact and convex set of matrices. By compactness, for each continuous function there exists a matrix $M\in \mathbb{M}$ that maximizes it.
This is true in particular for the continuous function
$U_w(M) := \sum_{i=1}^n w(M_{i,i})$. Every matrix $M\in \mathbb{M}$ that maximizes $U_w$ is a relative-value-matrix of an allocation $X$ that maximizes $\wrel_w$.

An analogous proof applies to $\wabs_w$.
\end{proof}
Since $w$ is strictly increasing, the \wmabs{w} and \wmrel{w} rules are trivially Pareto-optimal.

Probably the most famous rules in the family of welfare-maximizers are the \emph{utilitarian} rules, maximizing the sum of values. They are attained when $w$ is the identity function, $w(v)=v$.

As we will see in this section, the properties of welfare-maximizing rules crucially depend on the level of concavity of the function $w$. Specifically, we will be interested in whether $w$ is convex ($w'(x)$ is increasing) or concave ($w'(x)$ is decreasing). Additionally, we are interested in the following sub-class of the strictly concave functions:

\begin{definition}
\label{def:hyper-concave}
A differentiable function $w$ is \emph{hyper-concave} if $x w'(x)$ is decreasing.
\end{definition}
Note that for strictly increasing functions hyper-concavity implies strict-concavity,
but the opposite is not true\footnote{For example: $w(x)=\sqrt{x}$ is strictly concave but not hyper-concave.}.

In the following subsections we will examine the properties of rules from the two welfare-maximizer families. We will prove that:
\begin{itemize}
\item Both \wmabs{w} and \wmrel{w} rules are essentially-single-valued whenever $w$ is strictly-concave.
\item The \wmabs{w} rules are resource-monotonic whenever $w$ is concave.
\item The \wmrel{w} rules are proportional whenever $w$ is hyper-concave.
\end{itemize}
The link between concavity and proportionality is not surprising. Intuitively, if $w$ is strictly concave, then giving an additional unit of utility to an agent produces less and less social welfare as the agent's utility increases. For instance if we divide a cake among two agents, and there is a small slice that is worth approximately the same for the two agents, then a strictly concave $w$-maximizer will give that slice to the 'poorer' agent, while a strictly convex $w$-maximizer will give it to the 'richer' agent.

The link between concavity and resource-monotonicity seems more surprising at first.
The proof of Theorem \ref{thm:abs-wconcave-rm} provides some intuition about this relationship.


\subsection{Essentially-single-valuedness}
In this section we prove that strict concavity of $w$ is a sufficient and ``almost'' necessary condition for $w$-maximizer rules being essentially-single-valued.

\begin{theorem}
\label{thm:esv}

(a) If $w$ is strictly-concave, then both the \wmabs{w} and the \wmrel{w} are ESV.

(b) If $w$ is not strictly-concave, then the \wmabs{w} is not ESV.

(c) If $w$ is not strictly-concave in $[0,1]$, and the \wmrel{w} is proportional, then the \wmrel{w} is not ESV.
\end{theorem}

\begin{proof}[Proof of Theorem \ref{thm:esv}(a)]
Suppose $w$ is strictly-concave.
Define $\mathbb{M}$ as in the proof of Lemma \ref{lem:existence}; as mentioned there, $\mathbb{M}$ is a compact and convex set of matrices. By convexity, every strictly-concave function is maximized by a single matrix $M\in \mathbb{M}$.
This is particularly true for the function
$U_w(M):= \sum_{i=1}^n w(M_{i,i})$, which is strictly-concave since $w$ is.
Let $M^*$ be the unique matrix in $\mathbb{M}$ that maximizes $U_w$. Then, in every allocation that maximizes $\wrel_w$, the relative-value-matrix is $M^*$. Hence the \wmrel{w} is ESV.

An analogous proof applies to $\wabs_w$ and the \wmabs{w}.
\end{proof}
The proofs of parts (b) and (c) require the following technical lemma:
\begin{restatable}{lemma}{intervallemma}
\label{lem:interval}
Let $w$ be a twice-continuously-differentiable function.
\begin{enumerate}
  \item If $w$ is not strictly-concave in some interval $[a,b]$ with $a<b$, then $w$ is convex in some sub-interval $[s,t]$ with $a\leq s<t \leq b$.
  \item If $w$ is not concave in some interval $[a,b]$ with $a<b$, then $w$ is strictly-convex in some sub-interval $[s,t]$ with $a\leq s<t \leq b$.
\end{enumerate}
\end{restatable}
\begin{proof}
See Appendix.
\end{proof}
The maximum of a convex function in an interval is attained in an endpoint of the interval. Hence, when $w$ is convex in $[s,t]$, we can create scenarios where the welfare functions are maximized both when the utility-vector is $s,t$ and when it is $t,s$.
\begin{proof}[Proof of Theorem \ref{thm:esv}(b)]
Suppose $w$ is not strictly-concave and let $[s,t]$ be an interval in which it is convex, with $s<t$.
Consider the following cake:
	\begin{center}
		\begin{tabular}{l|c|c|cccc|c}
			\hline \c{$\vabs_A$}    & \c{$s$} & \c{0} & \c{$t-s$}  \\
			\hline \c{$\vabs_B$}    & \c{0} & \c{$s$} & \c{$t-s$}  \\
			\hline
		\end{tabular}
	\end{center}
By Pareto-optimality, slice \#1 goes to Alice and slice \#2 goes to Bob. Let $x$ be the fraction of slice \#3 that goes to Alice, so that Alice's value is $s(1-x)+tx$ and Bob's value is $sx+t(1-x)$. The social welfare as a function of $x$ is:
\begin{equation*}
	F(x)=w\bigg(s(1-x)+tx)\bigg)+w\bigg(sx+t(1-x))\bigg).
\end{equation*}
Since $w$ is convex on $[s,t]$, $F(x)$ is also convex, so its maximum is attained either at $x=0$ or at $x=1$. But $F(0)=F(1)=w(s)+w(t)$, so both must be maximum points. But, in these two maximum points, Alice and Bob receive different values (Alice receives more when $x=1$ and Bob receives more when $x=0$). Thus the \wmabs{w} is not ESV.
\end{proof}

\begin{proof}[Proof of Theorem \ref{thm:esv}(c)]
Suppose $w$ is not strictly-concave in $[0,1]$ and let $[s,t]$ be an interval in which it is convex, with $0\leq s<t\leq 1$.
Suppose also that the \wmrel{w} is proportional.
Consider the following cake, which has to be divided among $k+2$ agents --- Alice, Bob and the so called \emph{complementary agents} $C_1,\dots,C_{k}$ (where $k$ is sufficiently large such that ${2\over k}<t-s$):
	\begin{center}
		\begin{tabular}{l|c|c|cccc|c}
			\hline \c{$\vabs_A$}    & \c{$s$} & \c{0} & \c{$t-s-\frac{2}{k}$}& \c{$1-t+\frac{2}{k}$}  \\
			\hline \c{$\vabs_B$}    & \c{0} & \c{$s$} & \c{$t-s-\frac{2}{k}$}& \c{$1-t+\frac{2}{k}$}  \\
			\hline \c{$\vabs_{C_1}$} & \c{0} &\c{0}& \c{0} & \c{1}             \\
			\hline \c{...}          & \c{0} &\c{0}& \c{0} & \c{1}             \\
			\hline \c{$\vabs_{C_{k}}$} & \c{0} &\c{0}& \c{0} & \c{1}             \\
			\hline
		\end{tabular}
	\end{center}
By proportionality, each of the complementary agents must receive at least $1/(k+2)$ fraction of slice \#4. Thus
only a fraction ${2\over k+2}$ of this slice remains for Alice and Bob.
Both Alice and Bob value the remnants of slice \#4 as at most
\begin{align*}
\frac{2(1-t+\frac{2}{k})}{k+2}\le \frac{2}{k+2}<\frac{2}{k}.
\end{align*}
Let us denote slice \#3 and the remnants of slice \#4 by $H$.
Thus, both Alice and Bob value $H$ as $t-s-\varepsilon$, for some $\varepsilon>0$.

Fix a division of slice \#4 among the complementary agents. Let $W_C:=\sum_{i=1}^k w(\vrel_{C_i}(X_{C_i})) $ be the contribution of the complementary agents to the total welfare.
Note that $H$ is the only part of the remaining cake that both Alice and Bob find valuable, moreover both of them evaluate $H$ as $t-s-\epsilon$. Let $x \in [0,1]$ denote the proportion of $H$ that is given to Alice. Then, the social welfare as a function of $x$ is:

\begin{equation*}
F(x)=w\big(s+x(t-s-\varepsilon)\big)+w\big(s+(1-x)(t-s-\varepsilon)\big)+W_C.
\end{equation*}
Since $w$ is convex on $[s,t]$, $F(x)$ is also convex, so the maximum is attained either at $x=0$ or at $x=1$. But $F(0)=F(1)$, so both must be maximum points, implying that \wmrel{w} is not ESV.
\end{proof}

\begin{remark*}
We do not know whether Theorem \ref{thm:esv}(c) holds without the proportionality assumption, and whether it holds with a fixed number of agents.
\end{remark*}

\subsection{Resource-monotonicity}
In this section we prove that for absolute $w$-maximizers, weak concavity of $w$ is a necessary and sufficient condition for RM.

\begin{theorem}
\label{thm:abs-wconcave-rm}
The \wmabs{w} rule is resource-monotonic if and only if $w$ is concave.
\end{theorem}

\begin{proof}
\textbf{Only if part:}
We assume that $w$ is not concave and prove that the \wmabs{w} is not RM.
a	
By Lemma \ref{lem:interval} there is an interval in which $w$ is strictly convex.
Denote this interval by $[s,s+2t]$, for some $s\geq 0$ and $t>0$. Then, by convexity
	\begin{align*}
	&&
	w(s+2t)-w(s+t) &>w(s+t) -w(s)
	\\
	 \implies &&
	w(s) + w(s+2t) &>w(s+t) + w(s+t).
	\end{align*}
	The continuity of $w$ implies that there exists a small $\epsilon>0$ such that:
	\begin{align}
	\label{eq:w-strictly-convex}
	w(s) + w(s+2t) > w(s+t) + w(s+t+\epsilon).
	\end{align}
	Consider the following cake:
	\begin{center}
		\begin{tabular}{lccccl}
			&         &        &     & $\blacktriangledown$
			\\
			\hline
			$\vabs_A$ & \c{$s$} &  \c{0} & \c{$t+\epsilon$} & \c{\en 0}
			\\
			\hline
			$\vabs_B$ &  \c{0} &  \c{$s$} & \c{$t$} & \c{\en $t$}
			\\
			\hline&
		\end{tabular}
	\end{center}
	Initially the cake is made of only the three leftmost slices. By PO, Alice gets the first slice and Bob gets the second slice. For any $x\in[0,1]$ let $Y(x)$ be the allocation where Alice gets the first slice and $x$ fraction of the third slice. So Alice's absolute value is $s+x\cdot (t+\epsilon)$ and Bob's absolute value is $s+(1-x)\cdot (t)$. The absolute welfare, as a function of $x$, is:
	\begin{align*}
	F(x)=\wabs_w(Y(x)) = w(s+x\cdot (t+\epsilon)) + w(s+ (1-x)\cdot t)
	\end{align*}
When $x\in[0,1]$, the arguments to $w$ in the above expression are all in the range $[s,s+2t]$ in which $w$ is strictly convex. Hence, $F(x)$ is strictly convex in the interval $[0,1]$ so its maximum must be in one of the endpoints: $F(0)=w(s)+w(s+t)$ or $F(1)=w(s+t+\epsilon)+w(s)$. Since $w$ is increasing and $s+t+\epsilon>s+t$, the maximum is attained at $x=1$, so Alice gets all the slice and her value is $s+t+\epsilon$.
	
When the cake grows, PO dictates that the rightmost slice is given to Bob. Again, for any $x\in[0,1]$ let $Y(x)$ denote the allocation where Alice gets the first slice and $x$ fraction of the third slice. Now the absolute welfare is:
\begin{align*}
	G(x)=\wabs_w(Y(x)) = w(s+x\cdot (t+\epsilon)) + w(s+t + (1-x)\cdot t)
\end{align*}
the maximum in the range $[0,1]$ is either $G(0)=w(s)+w(s+2t)$ or $G(1)=w(s+t+\epsilon)+w(s+t)$. Inequality (\ref{eq:w-strictly-convex}) implies that the former is larger so the maximum is at $x=0$, Alice gets nothing from the third slice and her value drops to $s$. This proves that \wmabs{w} is not RM.

\textbf{If part:} we assume that $w$ is concave and prove that the \wmabs{w} is RM.

We start with upwards-RM. Let $X$ be an allocation that maximizes the welfare function $\wabs_w$ on the original cake $C$, and $Y$ an allocation that maximizes $\wabs_w$ on an enlarged cake $C\cup E$. We use the following definitions:
\begin{itemize}
	\item Agent $i$ is \emph{unlucky} if $\vabs_i(Y_i)<\vabs_i(X_i)$;
	\item Agent $j$ is \emph{lucky} if $\vabs_j(Y_j)>\vabs_j(X_j)$;
	\item Agent $k$ is \emph{indifferent} if $\vabs_k(Y_k)=\vabs_k(X_k)$.
	\item A pair of agents $(i,j)$ is a \emph{bad pair} if $i$ is unlucky, $j$ is lucky, and $i$ has conceded a positive slice to agent $j$, i.e, $X_i \cap Y_j$ is positive.
\end{itemize}
If there are no unlucky agents in $Y$, then we are done --- monotonicity is satisfied for all agents.

Moreover, if there are no bad pairs in $Y$, then by giving all unlucky and indifferent agents the original share that they had in $X$, the unlucky agents will be better-off and the other agents will not be harmed. Hence, if $Y$ is welfare-maximizing and has no bad pairs, then $Y$ has no unlucky agents at all, and we are done.

Therefore, to prove upwards-RM, it is sufficient to prove that there exists a division $Y'$ of $C\cup E$ with $\wabs_w(Y') \geq  \wabs_w(Y)$ where there are no bad pairs. Let $(i,j)$ be a bad pair and let $H:=X_i \cap Y_j$. By definition $H$ is positive, so by PO of $X$ and $Y$ it has a positive value for both $i$ and $j$. Let $z\in[0,1]$ be a number defined as:
\begin{align*}
z := \min\bigg(
\,\,\,\,\,{\vabs_i(X_i)-\vabs_i(Y_i) \over \vabs_i(H)}\,\,\,\,\,
,
\,\,\,\,\,{\vabs_j(Y_j)-\vabs_j(X_j) \over \vabs_j(H)}\,\,\,\,\,
,
\,\,\,\,\,1\,\,\,\,\,
\bigg)
\end{align*}
A theorem of \citet{Stromquist1985Sets} implies that there exists a subset $H^z\subseteq H$ such that:
\begin{align*}
\vabs_i(H^z) = z\cdot \vabs_i(H) && \text{and} && \vabs_j(H^z) = z\cdot \vabs_j(H)
\end{align*}
By definition of $z$, $H^z$ is sufficiently small such that:
\begin{align}
\vabs_i(Y_i \cup H^z)
=
\vabs_i(Y_i) + z\cdot \vabs_i(H)
&\le \vabs_i(X_i) \label{eq:Hsmall1}
\\
\vabs_j(Y_j\setminus H^z)
=
\vabs_j(Y_j) - z\cdot \vabs_j(H)
&\ge\vabs_j(X_j) \label{eq:Hsmall2}
\end{align}
both hold. That is, $i$ does not become lucky by getting $H^z$ and $j$ does not become unlucky by losing $H^z$ (though one of them may become indifferent).

The concavity of $w$, together with the inequalities (\ref{eq:Hsmall1}) and (\ref{eq:Hsmall2}), imply the following two inequalities:
\begin{align}
w(\vabs_i(Y_i\cup H^z))-w(\vabs_i(Y_i)) &\ge   w(\vabs_i(X_i))-w(\vabs_i(X_i\setminus H^z)); \label{eq:wcon3}
\\
w(\vabs_j(Y_j))-w(\vabs_j(Y_j\setminus H^z))
&\le
w(\vabs_j(X_j\cup H^z))-w(\vabs_j(X_j)).
\label{eq:wcon2}
\end{align}
By the optimality of $X$:
\begin{align}
&&
	w(\vabs_i(X_i))+w(\vabs_j(X_j)) &\ge w(\vabs_j(X_j\cup H^z))+w(\vabs_i(X_i\setminus H^z)) \notag
\\
\implies &&
	w(\vabs_i(X_i))-w(\vabs_i(X_i\setminus H^z))  &\ge w(\vabs_j(X_j\cup H^z))-w(\vabs_j(X_j)) \label{eq:wcon1}
\end{align}
Combining \ref{eq:wcon1}, \ref{eq:wcon2} and \ref{eq:wcon3} together yields
\begin{align*}
&&
	w(\vabs_i(Y_i\cup H^z))-w(\vabs_i(Y_i)) &\ge  w(\vabs_j(Y_j))-w(\vabs_j(Y_j\setminus H^z))
	\\
\implies &&
	w(\vabs_i(Y_i\cup H^z))+w(\vabs_j(Y_j\setminus H^z)) &\ge w(\vabs_j(Y_j))+w(\vabs_i(Y_i))
\end{align*}
So if we modify division $Y$ by transferring $H^z$ from $j$ back to $i$, the welfare weakly increases. Moreover, after the transfer, the pair $(i,j)$ is no longer a bad pair (either one of the agents becomes indifferent, or the conceded slice becomes empty). Moreover, no new bad pairs are created by the transfer, since no agents became lucky/unlucky. Therefore, we can remove the bad pairs one by one, until we get a new division $Y'$, which has at least the same welfare of $Y$ but no bad pairs. This implies that the absolute $w$-maximizer is upwards-RM.

For downwards-RM the proof is similar: here, $Y$ is the old division (of the larger cake) and $X$ is the new division (of the smaller cake). Define a bad pair as a pair $(i,j)$ such that:
\begin{itemize}
	\item Agent $i$ is \emph{lucky}, i.e, $\vabs_i(X_i)>\vabs_i(Y_i)$;
	\item Agent $j$ is \emph{unlucky}, i.e, $\vabs_j(X_j)<\vabs_j(Y_j)$;
	\item Agent $i$ has taken a positive slice from agent $i$, i.e, $X_i\cap Y_j$ is positive.
\end{itemize}
If there are no lucky agents, then we are done. Moreover, if there are no bad pairs, then the original division $Y$ could be modified by giving all lucky and indifferent agents the share they are going to receive in $X$ (since the cake of division $Y$ contains the cake of division $X$). This makes the lucky agents better-off and does not harm the unlucky agents. Hence, if $Y$ is welfare-maximizing and there are no bad pairs, then there must be no lucky agents at all, and we are done. From here, the proof that there exists a division $X'$ where there are no bad pairs follows the above proof word by word.
\end{proof}

\begin{remark*}
A relation between concavity and resource-monotonicity in the context of homogeneous goods is mentioned in Exercise 2.16 of \citet{Moulin2004Fair}.
\end{remark*}

\subsection{Population-monotonicity}
In this section we prove that, for both absolute and relative $w$-maximizers, weak concavity of $w$ is sufficient for PM.

\begin{theorem}\label{thm:pm}
If $w$ is concave then \wmabs{w} and \wmrel{w} are both population-monotonic.
\end{theorem}
\begin{proof}
When an agent joins or leaves, the total cake value does not change, so the proof for \wmabs{w} and \wmrel{w} is the same.

\ifdefined\ProofPM
We assume that $w$ is concave and prove that in this case the \wmabs{w} rule is PM.

We start with downwards-PM. Let $X$ be an allocation that maximizes the welfare function $\wabs_w$ among agents in $N$,
and $Y$ an allocation that maximizes $\wabs_w$ among agents in $N'\subsetneq N$.
We use the following definitions:
\begin{itemize}
	\item Agent $i\in N'$ is \emph{unlucky} if $\vabs_i(Y_i)<\vabs_i(X_i)$;
	\item Agent $j\in N'$ is \emph{lucky} if $\vabs_j(Y_j)>\vabs_j(X_j)$;
	\item Agent $k\in N'$ is \emph{indifferent} if $\vabs_k(Y_k)=\vabs_k(X_k)$.
	\item A pair of agents $(i,j)\in N'\times N'$ is a \emph{bad pair} if $i$ is unlucky, $j$ is lucky, and $i$ has conceded a positive slice to agent $j$, i.e, $X_i \cap Y_j$ is positive.
\end{itemize}
If there are no unlucky agents in $Y$, then we are done --- monotonicity is satisfied for all agents in $N'$.

Moreover, if there are no bad pairs in $Y$, then by giving all unlucky and indifferent agents in $N'$ the original share that they had in $X$, the unlucky agents will be better-off and the other agents will not be harmed. Hence, if $Y$ is welfare-maximizing and has no bad pairs, then $Y$ has no unlucky agents at all, and we are done.

Therefore, to prove downwards-PM, it is sufficient to prove that there exists a division $Y'$ among $n-1$ agents with $\wabs_w(Y') \geq  \wabs_w(Y)$ where there are no bad pairs. Let $(i,j)$ be a bad pair and let $H:=X_i \cap Y_j$. By definition $H$ is positive, so by PO of $X$ and $Y$ it has a positive value for both $i$ and $j$. Let $z\in[0,1]$ be a number defined as:
\begin{align*}
z := \min\bigg(
\,\,\,\,\,{\vabs_i(X_i)-\vabs_i(Y_i) \over \vabs_i(H)}\,\,\,\,\,
,
\,\,\,\,\,{\vabs_j(Y_j)-\vabs_j(X_j) \over \vabs_j(H)}\,\,\,\,\,
,
\,\,\,\,\,1\,\,\,\,\,
\bigg)
\end{align*}
A theorem of \citet{Stromquist1985Sets} implies that there exists a subset $H^z\subseteq H$ such that:
\begin{align*}
\vabs_i(H^z) = z\cdot \vabs_i(H) && \text{and} && \vabs_j(H^z) = z\cdot \vabs_j(H)
\end{align*}
By definition of $z$, $H^z$ is sufficiently small such that:
\begin{align}
\vabs_i(Y_i \cup H^z)
=
\vabs_i(Y_i) + z\cdot \vabs_i(H)
&\le \vabs_i(X_i) \label{eq:Hsmall1}
\\
\vabs_j(Y_j\setminus H^z)
=
\vabs_j(Y_j) - z\cdot \vabs_j(H)
&\ge\vabs_j(X_j) \label{eq:Hsmall2}
\end{align}
both hold. That is, $i$ does not become lucky by getting $H^z$ and $j$ does not become unlucky by losing $H^z$ (though one of them may become indifferent).

The concavity of $w$, together with the inequalities (\ref{eq:Hsmall1}) and (\ref{eq:Hsmall2}), imply the following two inequalities:
\begin{align}
w(\vabs_i(Y_i\cup H^z))-w(\vabs_i(Y_i)) &\ge   w(\vabs_i(X_i))-w(\vabs_i(X_i\setminus H^z)); \label{eq:wcon3}
\\
w(\vabs_j(Y_j))-w(\vabs_j(Y_j\setminus H^z))
&\le
w(\vabs_j(X_j\cup H^z))-w(\vabs_j(X_j)).
\label{eq:wcon2}
\end{align}
By the optimality of $X$:
\begin{align}
&&
	w(\vabs_i(X_i))+w(\vabs_j(X_j)) &\ge w(\vabs_j(X_j\cup H^z))+w(\vabs_i(X_i\setminus H^z)) \notag
\\
\implies &&
	w(\vabs_i(X_i))-w(\vabs_i(X_i\setminus H^z))  &\ge w(\vabs_j(X_j\cup H^z))-w(\vabs_j(X_j)) \label{eq:wcon1}
\end{align}
Combining \ref{eq:wcon1}, \ref{eq:wcon2} and \ref{eq:wcon3} together yields
\begin{align*}
&&
	w(\vabs_i(Y_i\cup H^z))-w(\vabs_i(Y_i)) &\ge  w(\vabs_j(Y_j))-w(\vabs_j(Y_j\setminus H^z))
	\\
\implies &&
	w(\vabs_i(Y_i\cup H^z))+w(\vabs_j(Y_j\setminus H^z)) &\ge w(\vabs_j(Y_j))+w(\vabs_i(Y_i))
\end{align*}
So if we modify division $Y$ by transferring $H^z$ from $j$ back to $i$, the welfare weakly increases. Moreover, after the transfer, the pair $(i,j)$ is no longer a bad pair (either one of the agents becomes indifferent, or the conceded slice becomes empty). Moreover, no new bad pairs are created by the transfer, since no agents became lucky/unlucky. Therefore, we can remove the bad pairs one by one, until we get a new division $Y'$, which has at least the same welfare of $Y$ but no bad pairs. This implies that the absolute $w$-maximizer is downwards-PM.

For upwards-PM the proof is similar.
Here $Y$ is the old division (among agents in $N'$) and $X$ is the new division (among agents in $N\supsetneq N'$).
Define a bad pair as a pair $(i,j)\in N'\times N'$ such that:
\begin{itemize}
	\item Agent $i$ is \emph{lucky}, i.e, $\vabs_i(X_i)>\vabs_i(Y_i)$;
	\item Agent $j$ is \emph{unlucky}, i.e, $\vabs_j(X_j)<\vabs_j(Y_j)$;
	\item Agent $i$ has taken a positive slice from $j$, i.e, $X_i\cap Y_j$ is a positive slice.
\end{itemize}
If there are no lucky agents in $X$, then we are done. Moreover, if there are no bad pairs, then the original division $Y$ could be modified by giving all lucky and indifferent agents in $N'$ the share they are going to receive in $X$. This makes the lucky agents better-off and does not harm the unlucky agents. Hence, if $Y$ is welfare-maximizing and there are no bad pairs, then there are no lucky agents at all, and we are done.
From here, the proof that there exists a division $X'$ where there are no bad pairs follows the above proof word by word.

\else
The proof of downwards-PM is almost identical to the proof of upwards-RM in the If Part of Theorem \ref{thm:abs-wconcave-rm}.
The only difference is that here $X$ is the original division among the agents in $N$,
and $Y$ is the new division among the agents in $N'\subsetneq N$.
Similarly, the proof of upwards-PM is almost identical to the proof of downwards-RM.
\fi
\end{proof}
Note that ESV cannot be used
to omit one of the upwards-PM or downwards-PM arguments,
since ESV is only guaranteed for strictly-concave functions (Theorem \ref{thm:esv}) while Theorem \ref{thm:pm} assumes merely concavity.
We do not know if concavity of $w$ is necessary for PM.

\subsection{Proportionality}
In this section we prove that, for relative $w$-maximizers, hyper-concavity (Definition \ref{def:hyper-concave}) is sufficient for proportionality. The main theorem is:
\begin{theorem}
	\label{lem:rel-prop}
	If the function $w$ is hyper-concave, then the \wmrel{w} rule is proportional.
\end{theorem}

To prove this theorem we use several lemmata about the properties of welfare-maximizing allocations. Although we need and prove them only for relative-$w$-maximizers, we note that they hold for absolute-$w$-maximizers too.
\begin{lemma}\label{lem:infi}
~\\
(a) Let $X$ be an allocation that maximizes the relative social-welfare function $\wrel_w$.
For every two agents $i,j$ and for every slice $H\subseteq X_j$:
\begin{align*}
w'(\vrel_j(X_j))\cdot \vrel_j(H)
\geq
w'(\vrel_i(X_i))\cdot \vrel_i(H)
\end{align*}
(so $H$ is given to an agent $j$ for whom the product $w'(\vrel_j(X_j)) \cdot \vrel_j(H)$ is maximal).
\\
(b) Let $X$ be an allocation that maximizes the absolute social-welfare function $\wabs_w$.
For every two agents $i,j$ and for every slice $H\subseteq X_j$:
\begin{align*}
w'(\vabs_j(X_j))\cdot \vabs_j(H)
\geq
w'(\vabs_i(X_i))\cdot \vabs_i(H)
\end{align*}
\end{lemma}
\begin{proof}
We prove part (a); the proof of (b) is entirely analogous.

We again use the theorem of \citet{Stromquist1985Sets}. For every $H$ and $z\in[0,1]$, there exists a subset $H^z\subseteq H$ such that:
\begin{align*}
\vrel_i(H^z) = z\cdot \vrel_i(H) && \text{and} && \vrel_j(H^z) = z\cdot \vrel_j(H)
\end{align*}
Let $Y(z)$ be an allocation derived from $X$ by taking $H^z$ from agent $j$ and giving it to agent $i$. The difference of welfare between the two allocations, as a function of $z$, is:
\begin{align*}
	F(z)&=\wrel_w(Y(z))-\wrel_w(X) =
	\\
	&= w\Big(\vrel_i(X_i) + \vrel_i(H^z)\Big)
	+ w\Big(\vrel_j(X_j) - \vrel_j(H^z)\Big) -w\Big(\vrel_i(X_i)\Big)-w\Big(\vrel_j(X_j)\Big)
	\\
	&= w\Big(\vrel_i(X_i) + z\cdot \vrel_i(H)\Big)
	+ w\Big(\vrel_j(X_j) - z\cdot \vrel_j(H)\Big) -w\Big(\vrel_i(X_i)\Big)-w\Big(\vrel_j(X_j)\Big)
\end{align*}
	Take the derivative as a function of $z$:
\begin{align*}
	F'(z) &= w'\Big(\vrel_i(X_i) + z\cdot \vrel_i(H)\Big)\cdot \vrel_i(H)
	- w'\Big(\vrel_j(X_j) - z\cdot \vrel_j(H)\Big)\cdot \vrel_j(H)
\end{align*}
	When $z=0$, the alternative allocation $Y(z)$ is identical to the original allocation $X$, and we know that this allocation maximizes $\wrel_w$, so 0 is a maximum point of $F$.  Therefore, $F'(0)\leq 0$:
\begin{align*}
&&
0\geq F'(0) = w'\Big(v_i(X_i)\Big)\cdot v_i(H) - w'\Big(v_j(X_j)\Big)\cdot v_j(H)
\\
\implies &&
w'\Big(v_j(X_j)\Big)\cdot v_j(H)
\geq
w'\Big(v_i(X_i)\Big)\cdot v_i(H).
\end{align*}
\end{proof}

We will now prove an interesting property of welfare-maximization with hyper-concave functions:
a poor agent never envies a richer agent. Formally, given an allocation $X$, we say that:
\begin{itemize}
	\item An agent $i$ \emph{envies} agent $j$, if $\vrel_i(X_i)<\vrel_i(X_j)$.
	Note that values of the same agent are compared.
	\item An agent $i$ is \emph{relatively/absolutely richer} than agent $j$, if $\vrel_i(X_i)>\vrel_j(X_j)$  / $\vabs_i(X_i)>\vabs_j(X_j)$.
	Note that values of different agents are compared.
	\item An agent $i$ is \emph{relatively/absolutely poorer} than agent $j$, if $\vrel_i(X_i)<\vrel_j(X_j)$  / $\vabs_i(X_i)<\vabs_j(X_j)$.
\end{itemize}

\begin{lemma}\label{lem:envy}
Let $w$ be a hyper-concave function ($x w'(x)$ is weakly-decreasing).

(a) If, in an allocation selected by the \wmrel{w} rule, an agent $i$ envies an agent $j$, then agent $i$ is relatively-richer than agent $j$.

(b) If, in an allocation selected by the \wmabs{w} rule, an agent $i$ envies an agent $j$, then agent $i$ is absolutely-richer than agent $j$.
\end{lemma}
\begin{proof}
We prove part (a); the proof of (b) is entirely analogous.
By Lemma \ref{lem:infi} (taking $H=X_j$):
\begin{align*}
	w'(\vrel_j(X_j))\cdot \vrel_j(X_j)
	&\geq
	w'(\vrel_i(X_i))\cdot \vrel_i(X_j)
\end{align*}

Combining the latter inequality with the assumption that $i$ envies $j$ ($v_i(X_j) > v_i(X_i)$) gives:
\begin{align*}
	w'(\vrel_j(X_j))\cdot \vrel_j(X_j)
	&>
	w'(\vrel_i(X_i))\cdot \vrel_i(X_i)
\end{align*}
Since $x w'(x)$ is weakly decreasing, this implies:
\begin{align*}
	\vrel_j(X_j) < \vrel_i(X_i)
\end{align*}
so $i$ is relatively-richer than $j$.
\end{proof}

Now we can prove our main theorem.
\begin{proof}[Proof of Theorem \ref{lem:rel-prop}]
Let $X$ be an allocation selected by the \wmrel{w} rule.
We prove that $X$ is proportional.

Call an agent $i$ ``unhappy'' if $\vrel_i(X_i)<1/n$.
Suppose by contradiction that $X$ is not proportional. Then there is at least one unhappy agent, say $i$. By the pigeonhole principle, $i$ necessarily envies some other agent, say $j$.
	By Lemma \ref{lem:envy}, agent $j$ must be relatively-poorer than $i$, so $j$ is also unhappy.
		Now, consider the set of all unhappy agents. Since each agent in the set envies another agent in the set, there must be a cycle of agents envying each other. But this contradicts the optimality of $X$. Hence, the set of unhappy agents must be empty.
\end{proof}

\begin{remark}
We do not know if hyper-concavity is necessary for proportionality. We do know that
strict concavity is not sufficient for proportionality.
For example, for some constant $p\in(0,1)$, Let $w_p(x) = x^p$. Note that $w_p$ is strictly concave but not hyper-concave. Consider the following cake:
\begin{center}
\begin{tabular}{lccccl}
&         &
\\
\hline
$\vabs_A$ & \c{1} &  \c{0}
\\
\hline
$\vabs_{B}$ &  \c{2/3} &  \c{1/3}
\\
\hline&
\end{tabular}
\end{center}
Let $x\in[0,1]$ be the value given to Alice. The value remaining for Bob is $1-2x/3$. By proportionality, Bob must receive at least $1/2$, so a proportional rule must select $x\leq 3/4$. The total relative (and absolute) welfare, as a function of $x$, is given by:
\begin{align*}
&&
F(x) &= w_p(x) + w_p\bigg(1-{2x\over 3}\bigg)
\\
\implies &&
F'(x) &= w_p'(x) - {2\over 3} w_p'\bigg(1-{2x\over 3}\bigg)
\\
\implies &&
{3\over 4} F'\bigg({3\over 4}\bigg) &= {3\over 4}w_p'\bigg({3\over 4}\bigg) - {1\over 2} w_p'\bigg({1\over 2}\bigg)
\end{align*}
The latter expression is positive since
$x w_p'(x)$ is increasing. Hence $F'(3/4)>0$. Since $w_p'$ is decreasing, $F'$ is also decreasing, so $F'(x)>0$ for all $x\leq 3/4$. Hence $F$ cannot have a maximum point at $x\leq 3/4$, so  the \wmrel{w} rule is not proportional. \qed
\end{remark}

\subsection{The Nash-optimal rule}
Let us collect our findings so far.
\begin{itemize}
\item The \wmabs{w} is RM and PM when $w$ is concave, and ESV when $w$ is strictly-concave.
\item The \wmrel{w} is PM when $w$ is concave, ESV when $w$ is strictly-concave, and PROP when $w$ is hyper-concave.
\end{itemize}
Therefore, to get all desirable properties simultaneously, we have to find a rule that is simultaneously maximizing the absolute and relative social welfare with the same hyper-concave function $w$.

Indeed, such a rule exists. When $w$ is a logarithmic function, \wmabs{w} and \wmrel{w} are both equivalent to the \emph{Nash-optimal rule} --- the rule that maximizes the product of utilities \citep{Nash1950Bargaining}:
\begin{align*}
X_{Nash} := \arg\max_{X} \prod_{i=1}^{n} \vabs_i(X_i)
 := \arg\max_{X} \prod_{i=1}^{n} \vrel_i(X_i)
\end{align*}

\begin{corollary}
\label{cor:nash}
The Nash-optimal rule is ESV, RM, PM and PROP.
\end{corollary}
Is Nash-optimal the only welfare-maximizing rule with these four properties? Below we prove that it is indeed unique.
\begin{theorem}
\label{thm:uniqueness}
(a) If the \wmabs{w} is proportional, then it must be the Nash-optimal rule.

(b) If the \wmrel{w} rule is ESV, proportional and resource monotonic, then it must be the Nash-optimal rule.
\end{theorem}

We start by proving uniqueness in the family of absolute welfare maximizers.
\begin{proof}[Proof of Theorem \ref{thm:uniqueness}(a)]
	Consider the following one-slice cake:
	\begin{center}
		\begin{tabular}{lc}
			\hline
			$\vabs_A$ & \c{$2a$}
			\\
			\hline
			$\vabs_B$ &  \c{$2b$}
			\\
			\hline&
		\end{tabular}
	\end{center}
A proportional allocation must give each agent exactly half the cake, so that $\vabs_A(X_A)=\vabs_A(X_B)=a$ and $\vabs_B(X_B)=\vabs_B(X_A)=b$.

Using Lemma \ref{lem:infi} with $j=$Alice and $H=X_A$ and $i=$Bob gives:
\begin{align*}
w'(a)\cdot a \geq w'(b)\cdot b
\end{align*}
Using Lemma \ref{lem:infi} with $j=$Bob and $H=X_B$ and $i=$Alice gives:
\begin{align*}
w'(b)\cdot b \geq w'(a)\cdot a
\end{align*}
Since $a,b$ are general, this implies that:
\begin{align*}
\forall a,b: && a w'(a) &= b w'(b)
\end{align*}
This means that the function $a w'(a)$ is a constant function (independent of $a$).
Hence, $w$ must be a logarithmic function ($w(\cdot)=c\ln{(\cdot)}+d$, for some constants $c>0$ and $d$) so \wmabs{w} is the Nash-optimal rule.
\end{proof}

We now turn to proving the uniqueness of the Nash-optimal rule in the family of relative
 welfare-maximizers.

\begin{proof}[Proof of Theorem \ref{thm:uniqueness}(b)]
By Theorem \ref{thm:esv}(c), if \wmrel{w} is PROP and ESV then $w$ is strictly concave in $[0,1]$. Hence our theorem follows from the next lemma.
\end{proof}

\begin{lemma}
	\label{lem:rel-not-rm}
	If $w$ is strictly concave in $[0,1]$ and the \wmrel{w} rule is proportional and resource-monotonic, then it is the Nash-optimal rule.
\end{lemma}
\begin{proof}[Proof sketch] We consider a cake that has to be divided among $k+2$ agents --- Alice, Bob and the so called \emph{complementary agents} $C_1,\dots,C_k$. We assume that $k$ is large. One slice of the cake --- the 'disputed slice' --- is wanted only by Alice and Bob, and the main task of the \wmrel{w} rule is to decide how this slice is divided between them.
	
In the initial situation, the cake is small, and both Alice and Bob value the entire cake as 1. Their value measures are similar, so a strictly concave rule must give each of them exactly 1/2 of the disputed slice. Then the cake grows. The enlargement is valuable only for Bob and for the complementary agents, but not for Alice. When $k$ is sufficiently large, the complementary agents take all the enlargement, so Bob gains no value from it; the only effect of the enlargement is that Bob's value for the entire cake is larger, so Bob's relative value for the disputed slice is smaller. This  breaks the symmetry and causes the \wmrel{w} rule to give either Alice or Bob a smaller share of the disputed slice --- in contradiction to resource-monotonicity. The only case in which this does not happen is when $w$ is a logarithmic function, which implies that \wmrel{w} is the Nash-optimal rule.
	
All the above has to be done twice: once to prove that $w(x)$ is logarithmic when $x\in [0,1/2]$, and then to prove that it is logarithmic when $x\in [1/2,1]$.
\end{proof}
\begin{proof}
	In the following cakes, the valuations are parameterized by $s$ and $t$.
	
\textbf{Cake 1} has the following valuations, for $0<s<t<1$ and $t\geq 1/2$:
	
\begin{center}
	\begin{tabular}{lcccccc}
		&    &    &       &          & $\blacktriangledown$\\
		\hline \c{$\vabs_A$} & \c{$2-2t$} & \c{$2t-1$} & \c{0}  & \c{0} & \c{\en 0}  \\
		\hline \c{$\vabs_B$} & \c{$2-2t$} & \c{0} & \c{$2t-1$} & \c{0} & \c{\en  ${t / s}-1$}    \\
		\hline \c{$\vabs_{C1}$} & \c{0} & \c{0} & \c{0} &\c{1}& \c{\en  ${t / s}-1$}             \\
		\hline \c{...}           & \c{0} & \c{0}& \c{0} &\c{1}& \c{\en  ${t / s}-1$}             \\
		\hline \c{$\vabs_{Ck}$}  & \c{0} & \c{0}& \c{0} &\c{1}& \c{\en  ${t / s}-1$}             \\
		\hline
	\end{tabular}
\end{center}
	
Initially the cake contains only the four leftmost slices. Alice gets slice \#2 and Bob gets slice \#3. By strict concavity, slice \#1 is divided equally between Alice and Bob, and their value (relative and absolute alike) is $t$.

When the cake grows, the cake value increases from 1 to $t/s$ for Bob and for the complementary agents. Let $Y$ be the new allocation. The relative value of Bob is now at least $(2t-1)/(t/s)$. In contrast, the complementary agents have to share a relative value of 1, so there is at least one agent for whom: $\vrel_{Ci}(Y_{Ci})\leq 1/k$. Therefore, when $k$ is sufficiently large:
\begin{align}
\vrel_{B}(Y_{B}) &> \vrel_{Ci}(Y_{Ci})
\label{eq:w'b<w'ci}
\end{align}
We claim that, for such $k$, Bob does not receive any value from the enlargement. Indeed, suppose by contradiction that Bob receives from the enlargement a slice $H$ with positive value. Since $H \subseteq Y_{B}$, by Lemma \ref{lem:infi}:
\begin{align*}
w'(\vrel_{B}(Y_{B}))\cdot {\vrel_B(H)}
\geq
w'(\vrel_{Ci}(Y_{Ci}))\cdot {\vrel_{C_i}(H)}
\end{align*}
Since $H$ is a subset of the enlargement, $\vrel_B(H)=\vrel_{C_i}(H)$, and by assumption this value is positive, so the above implies:
\begin{align*}
w'(\vrel_{B}(Y_{B}))
\geq
w'(\vrel_{Ci}(Y_{Ci}))
\end{align*}
But this combined with (\ref{eq:w'b<w'ci}) contradicts the strict concavity of $w$.

Since --- for sufficiently large $k$ --- Bob does not receive anything from the enlargement, by resource-monotonicity he must receive at least half of slice \#1. The same is true for Alice, so both of them must receive exactly half of slice \#1. Thus each of them has an absolute value of exactly $t$, so their relative values are:
\begin{align*}
\vrel_A(Y_A)=t/1 = t && \text{and} && \vrel_B(Y_B)=t/(t/s)=s
\end{align*}
Apply Lemma \ref{lem:infi} with $j=$Alice, $i=$Bob and $H=$Alice's share of slice \#1:
\begin{align*}
w'(t)\cdot (1-t) \geq w'(s)\cdot {1-t \over t/s}
\end{align*}
Apply Lemma \ref{lem:infi} with $j=$Bob, $i=$Alice and $H=$Bob's share of slice \#1:
\begin{align*}
w'(s)\cdot {1-t \over t/s} \geq w'(t)\cdot (1-t)
\end{align*}
Combine these two inequalities to obtain $t w'(t) = s w'(s)$, for any $t\geq 1/2$. Formally,
\begin{align}
\label{eq:tgeq1/2}
\forall t\geq 1/2: \forall s \text{ with } 0<s<t<1:
t w'(t) = s w'(s)
\end{align}

	\textbf{Cake 2} has the following valuations, for $0<s<t<1$ and $t\leq 1/2$:
	
	\begin{center}
		\begin{tabular}{lccccc}
			&       &     &     & & $\blacktriangledown$\\
			\hline \c{$\vabs_A$}    & \c{$2t$} & \c{$0$} & \c{$1-2t$} &  \c{0}&  \c{\en  0}         \\
			\hline \c{$\vabs_B$}    & \c{$2t$} & \c{$0$} & \c{$1-2t$} &  \c{0}& \c{\en  ${t / s}-1$}\\
			\hline \c{$\vabs_{C_1}$} & \c{0} & \c{$2t$} &\c{$1-2t$}&  \c{0}& \c{\en  0}         \\
			\hline \c{...}          & \c{0} & \c{$2t$} &\c{$1-2t$}&  \c{0}& \c{\en  0}         \\
			\hline \c{$\vabs_{C_k}$} & \c{0} & \c{$2t$} &\c{$1-2t$}&  \c{0}& \c{\en  0}         \\
			\hline \c{$\vabs_{C_{k+1}}$} & \c{0} & \c{$0$} &\c{0}&  \c{$1$}& \c{\en  $t/s-1$}         \\
			\hline \c{...}          & \c{0} & \c{$0$} &\c{0}&  \c{$1$}& \c{\en  $t/s-1$}         \\
			\hline \c{$\vabs_{C_{2k}}$} & \c{0} & \c{$0$} &\c{0}&  \c{$1$}& \c{\en  $t/s-1$}         \\
			\hline
		\end{tabular}
	\end{center}
	Initially the cake contains only the four leftmost slices. By strict concavity Alice and Bob divide the slices valuable to both of them equally, so each of them receives at least $t$.
	
When the cake grows, Bob's total cake value as well as the second $k$-set of complementary agent's cake value increases from 1 to $t/s$. By the same reasoning as in case of Cake 1, Bob does not receive anything from the enlargement, and neither Alice nor Bob receive anything from slice \#3. By resource-monotonicity, Alice and Bob should have an absolute value of at least $t$, so they must split slice \#1 equally, giving each of them an absolute value of exactly $t$. Hence, $\vrel_A(Y_A)=t$ and $\vrel_B(Y_B)=t/(t/s)=s$. From here, the proof is the same as in the previous case, and we get:
\begin{align}
\label{eq:tleq1/2}
\forall t\leq 1/2: \forall s \text{ with } 0<s<t<1:
t w'(t) = s w'(s)
\end{align}

The equations (\ref{eq:tgeq1/2}) and (\ref{eq:tleq1/2}) together imply that the function $t w'(t)$ must be the constant function for all $t\in(0,1)$. This implies that $w$ is a logarithmic function, hence the \wmrel{w} is the Nash-optimal rule.	
\end{proof}

\begin{remark}
The example in Lemma \ref{lem:rel-not-rm} involves a possibly unbounded number of agents. This raises the following open question: what division rules are PROP and RM when the number of agents is bounded by some constant?
\end{remark}

From the two parts of Theorem \ref{thm:uniqueness}, we get:
\begin{corollary}
	In the family of welfare-maximizers, the Nash-optimal rule is the only essentially-single-valued rule that is both PROP and RM.
\end{corollary}

\subsection{Summary of division-rule properties}
To summarize the properties of division rules proved in this section, Table \ref{tab:summary2} presents the properties that are satisfied by welfare-maximizing rules from a well-studied one-parametric family \citep[chapter 3]{Moulin2004Fair}:
\begin{align*}
w_p(x) = x^p / p && \text{when $p\neq 0$}
\\
w_p(x) = \ln(x) && \text{when $p = 0$}
\end{align*}

\begin{table}[h!]
\footnotesize\centering
	\begin{tabular}{|c|c|c|}
		\hline Axiom & \wmabs{w_p} & \wmrel{w_p} \\
		\hline PO & All & All \\
		\hline ESV & $p<1$ (strictly concave)&$p<1$ (strictly concave) \\
		\hline PROP &  $p=0$ (Nash) &  $p\le0$ (hyper-concave) \\
		\hline RM & $p\le1$ (concave)& $p=0$ (Nash) \\
		\hline PM & $p\le1$ (concave) & $p\le1$ (concave)  \\
		\hline
	\end{tabular}
\protect\caption{Sufficient conditions for the parameterized welfare-maximizing rules.}\label{tab:summary2}
\end{table}

\subsection{Computing Nash-optimal allocations}
\label{sub:nash-computation}
The nice properties of the Nash-optimal rule motivate us to search for algorithms for finding the Nash-optimal division.
We are not aware of an  algorithm for finding an exact Nash-optimal division in the general case\footnote{
There are several algorithms for calculating or approximating the Nash-optimal welfare in markets with divisible goods
\citep{Cole2013Mechanism,Branzei2017Nash}
or indivisible goods \citep{Caragiannis2016Unreasonable}
}.
However,
\citet{Aziz2014Cake} present an algorithm that finds a Nash-optimal division when all agents have piecewise-constant valuations.\footnote{
Their algorithm is based on
earlier algorithms
for finding equilibrium allocations in markets of homogeneous divisible goods, which are equivalent to piecewise-homogeneous cakes. See
\citet{vazirani2007combinatorial,jain2010eisenberg}.
}
The algorithm is called MEA (Market Equilibrium Algorithm)
and it is based on a connection between the Nash-optimal rule and a well-known division rule called \emph{competitive-equilibrium-from-equal-incomes}.

In the following section we prove such connection for general valuations (not only piecewise-constant).
This will imply that any future algorithm developed for one of the rules will work for the other rule as well.

\section{Competitive Equilibrium from Equal Incomes}
\label{sec:ceei-nash}
\subsection{CEEI rules in cake-cutting}
Competitive Equilibrium from Equal Incomes (CEEI) is a well-known rule for fair and efficient allocation of homogeneous goods. It was first introduced into the cake-cutting world by \cite{Weller1985Fair}.
We denote Weller's definition by \emph{WCEEI (weak CEEI)}; the reason will become clear later.
\begin{definition}\label{def:ceei}
Let $C$ be a cake and $X$ an allocation of the cake. Let $P$ be a nonatomic measure on $C$ (called the "price measure"). We say that the pair $(X,P)$ is a \emph{Weak Competitive Equilibrium from Equal Incomes (WCEEI)}  if it satisfies the following conditions:
	\begin{itemize}
		\item WCE: For all $i\in N$
		and $Z\subset C$,
		$P(Z)\le P(X_i)$ implies $\vrel_i(Z) \leq \vrel_i(X_i)$.
		\item EI: For all $i\in N$: $P(X_i)=1$.
	\end{itemize}
\end{definition}
Note that it does not matter whether absolute or relative values are used in this definition, since only values of the same agent are compared.

\citet{Weller1985Fair} proved that a WCEEI cake-allocation always exists. Moreover, a WCEEI allocation has several nice properties.
\begin{lemma}[\citet{Weller1985Fair}]
\label{lem:wceei-ef}
Every WCEEI allocation is envy-free (hence also proportional).
\end{lemma}
\begin{proof}
WCE implies that each agent $i$ prefers his piece $X_i$ over all pieces he can afford. EI implies that all agents have the same set of affordable pieces. Together they imply envy-freeness.
\end{proof}

\begin{lemma}
\label{lem:wceei-wpo}
Every WCE allocation is weakly-Pareto-optimal, i.e, no other allocation is strictly better for all agents.
\end{lemma}
\begin{proof}
Suppose by contradiction that $(X,P)$ satisfies the WCE condition but there exists an allocation $Y$ which is strictly better than $X$ for all agents: $\forall i: v_i(Y_i)>v_i(X_i)$. The WCE condition implies that no agent $i$ can afford the piece $Y_i$: $\forall i: P(Y_i) > P(X_i)$. Summing over all pieces gives  $\sum_{i=1}^n P(Y_i) > \sum_{i=1}^n P(X_i)$. But this is impossible since both sides equal $P(C)$.
\end{proof}
Interestingly, the WCEEI rule is not Pareto-optimal.
\footnote{
Note that Weller's definition of  Pareto-optimality (before his Theorem 1)  actually defines weak-Pareto-optimality.
We are grateful to an anonymous reviewer for this comment.}.
\begin{example}
\label{exm:wceei-not-po}
Consider the following cake.
\begin{center}
	\begin{tabular}{|l|c|c|c|c|c|c|}
		\hline  $\vabs_A$ & \cellcolor{myGreen!25}30 & \cellcolor{myGreen!25}0 & 0 \\
		\hline  $\vabs_B$ & 0 & 10 & \cellcolor{myBlue!25}20  \\
		\hline
		\hline  Price & 0.2 & 0.8 & 1  \\
		\hline
	\end{tabular}
\end{center}
The two leftmost slices go to Alice. She pays $1$ and gets all her value, so the WCEEI conditions are satisfied for her.
The rightmost slice goes to Bob. He pays $1$ and cannot get any better deal for the same price: if he sells a fraction $x$ of his slice, then he loses a value of $20 x$ and gains $x$ money. With $x$ money he can only buy $x/0.8 = 1.25 x$ of the middle slice, which gives him only $12.5 x$ value.
So the WCEEI conditions are satisfied for Bob too. However, it is clear that the allocation is not PO since it can be Pareto-improved by giving the middle slice to Bob.

\end{example}
Pareto-optimality can be restored by adding a requirement that, when an agent is indifferent between several best slices, he gets the cheapest of these slices.
We call the strengthened condition \emph{Parsimonious Competitive Equilibrium (PCE).}%
\footnote{
The condition was introduced by \citet{mas1992equilibrium} and termed ``parsimony'' by
\citet{Bogomolnaia2017}.
}
It requires that for all $i\in N$ and $Z\subset C$:
\begin{itemize}
\item $P(Z)\le P(X_i)$ implies $\vrel_i(Z) \leq \vrel_i(X_i)$ (the WCE condition), and:
\item $P(Z) < P(X_i)$ implies $\vrel_i(Z) <  \vrel_i(X_i)$.
\end{itemize}
\begin{lemma}
\label{lem:pceei-ppo}
Every PCE allocation is Pareto-optimal.
\end{lemma}
\begin{proof}
Suppose by contradiction that $(X,P)$ satisfies the PCE condition but there exists an allocation $Y$ which is weakly better for all agents $i$ and strictly better for at least one agent $j$. Then, the PCE conditions imply that $\forall i: P(Y_i)\geq P(X_i)$ and $P(Y_j)>P(X_j)$. Summing over all pieces gives  $\sum_{i=1}^n P(Y_i) > \sum_{i=1}^n P(X_i)$. But this is impossible since both sides equal $P(C)$.
\end{proof}
However, the PCEEI allocation is still not essentially-single-valued.
\begin{example}
\label{exm:pceei-not-esv}
Consider the following cake.
\begin{center}
	\begin{tabular}{|l|c|c|c|c|c|c|}
		\hline  $\vabs_A$ & \cellcolor{myGreen!25}30 & \cellcolor{myGreen!25}1 & 0 \\
		\hline  $\vabs_B$ & 0 & \cellcolor{myBlue!25}10 & \cellcolor{myBlue!25}20  \\
		\hline
		\hline  $P_1$ & 0.2 & 0.8 & 1  \\
		\hline  $P_2$ & 1   & 0.8 & 0.2  \\
		\hline
	\end{tabular}
\end{center}
Consider the following two allocations: (1) Alice gets the two leftmost slices and Bob gets the rightmost slice and the price is $P_1$, or (2) Alice gets the leftmost slice and Bob gets the two rightmost slices and the price is $P_2$.
Both of them are PCEEI but the agents' valuations are different. \qed
\end{example}


We now define an even stronger condition which we call \emph{SCEEI}. It is adapted from \citet{Reijnierse1998Finding}.
\begin{definition}\label{def:SCEei}
Let $X$ be a cake-allocation and $P$ be a measure on $C$.
The pair $(X,P)$ is a \emph{strong CEEI} (\emph{SCEEI}) if it satisfies the following conditions:
\begin{itemize}
\item
$P(Z) > 0$ iff $Z$ is a positive slice (= valued positively by at least one agent).
\item SCE: For every agent $i$
and positive-slice $Z\subseteq C$
and  positive-slice $Z_i\subseteq X_i$:
$\vrel_i(Z_i)/P(Z_i) \geq \vrel_i(Z)/P(Z)$.
In words, each agent buys only slices that maximize his value-per-price ratio.
\item EI: For every agent $i$: $P(X_i)=1$.
\end{itemize}
\end{definition}
The following lemma justifies the term SCEEI:
\begin{lemma}
Every SCEEI is a PCEEI (hence also PO, EF and PROP).
\label{lem:sceei-is-wceei}
\end{lemma}
\begin{proof}
Since the EI condition is the same in all equilibrium variants, it is sufficient to prove the PCE conditions for every agent $i$.

The EI condition implies that $P(X_i)>0$, so $X_i$ is a positive slice.  Setting $Z_i := X_i$ in the SCE condition implies that, for every positive slice $Z$:
\begin{align*}
\vrel_i(Z)/P(Z) \leq \vrel_i(X_i) / P(X_i)
\end{align*}
Hence, for every positive slice $Z$:
\begin{itemize}
\item $P(Z)\leq P(X_i)$ implies $\vrel_i(Z)\leq \vrel_i(X_i)$ (the WCE condition).
\item $P(Z)<P(X_i)$ implies $\vrel_i(Z)< \vrel_i(X_i)$ (the extra PCE condition).
\qedhere
\end{itemize}
\end{proof}

\subsection{Properties of the the Strong CEEI rule}
In this subsection we prove that the SCEEI rule has many nice properties: in addition to being PO and EF, it is also ESV, RM and PM.
We will use a measure $V$, defined as the sum of all agents' value-measures:
\begin{align*}
V(Z) := \sum_{i\in N} v_i(Z)
\end{align*}
Let $(X,P)$ be a SCEEI. By definition, $P(Z)>0$ iff $v_i(Z)>0$ for at least one $i$, which holds iff $V(Z)>0$. Therefore, $P$ is absolutely-continuous w.r.t. $V$.
By the Radon-Nikodym theorem, there exists a \emph{price-density}  function $p$, such that for every slice $Z\subseteq C$:
\begin{align*}
P(Z)=\int_{x\in Z}{p(x) dV}
\end{align*}

Let $(Y,Q)$ be a SCEEI on an enlarged cake $C \cup E$, and let $q$ be the price-density of $Q$, such that for every slice $Z\subseteq C\cup E$:
\begin{align*}
Q(Z)=\int_{x\in Z}{q(x) dV}
\end{align*}

Define the following subset of the original cake $C$, which contains all those parts of $C$ that are more expensive in equilibrium $Y$ than in equilibrium $X$:

$$C^* = \{x\in C \ | \ q(x)>p(x)\}$$
For every positive-slice $Z^*\subseteq C^*$: $Q(Z^*)>P(Z^*)$. By definition, for every $x\in (C\setminus C^*)$, $q(x)\leq p(x)$. Thus, for every slice $Z'\subseteq (C\setminus C^*)$: $Q(Z')\leq P(Z')$.

\begin{lemma}\label{lemma:CEEI-claim}
If in equilibrium $Y$ agent $i$ holds a positive subset of $C^*$ (i.e.\ $Y_i\cap C^*$ is a positive-slice), then in equilibrium $X$ agent $i$ holds almost only subsets of $C^*$ ($X_i\cap (C\setminus C^*)$ is not a positive-slice).
\end{lemma}
\begin{proof}
Define $Z^*_i = Y_i\cap C^*$.
By assumption it is a positive slice. Because $Z^*_i\subseteq Y_i$, the SCE condition of equilibrium $Y$ implies that for every positive-slice $Z'\subseteq C\setminus C^*$:
	
$$v_i(Z^*_i)/Q(Z^*_i) \geq v_i(Z')/Q(Z')$$
Because $Z^*_i\subseteq C^*$ and $Z'\subseteq C\setminus C^*$, the definition of $C^*$ implies that the old price of $Z^*_i$ was lower than the new price and the old price of $Z'$ was weakly higher than the new price. Hence,
	
$$v_i(Z^*_i)/P(Z^*_i) > v_i(Z')/P(Z')$$
Now the SCE condition of equilibrium $X$ implies that $Z'$ is not contained in $X_i$.
	
The above is true for every $Z'$ which is a positive slice of $C\setminus C^*$; hence, $X_i \cap (C\setminus C^*)$ cannot be a positive slice.
\end{proof}

The following lemma says that, when the cake grows, the equilibrium price of the old cake does not increase.

\begin{lemma}
	\label{ceei-decreasing-price}
	Let $(X,P)$ be a SCEEI on the cake $C$ and $(Y,Q)$ a SCEEI on an enlarged cake $C\cup E$. Then, for every subset of the original cake $Z\subseteq C$: $Q(Z)\leq P(Z)$.
\end{lemma}
\begin{proof}
	Suppose there are exactly $k$ agents that hold a positive slice of $C^*$ in equilibrium $Y$. Their total income is $k$ and they can afford all of  $C^*$, so:
	
	$$k\geq Q(C^*)$$
	On the other hand, by Lemma \ref{lemma:CEEI-claim}, in Equilibrium $X$ these $k$ agents spend their entire income on $C^*$, so:
	
	$$k\leq P(C^*)$$
	Combining these two inequalities gives:
	
	$$P(C^*)\geq Q(C^*)$$
	But, for every positive-slice $Z^*\subseteq C^*$, $Q(Z^*)>P(Z^*)$. We conclude that $C^*$ cannot be a positive slice, i.e, $V(C^*)=0$. Hence, $q(x)\leq p(x)$ almost everywhere (w.r.t. $V$). This implies the lemma.
\end{proof}

\begin{corollary}
	\label{ceei-same-price}
	Let $(X,P)$ and $(Y,Q)$ be two SCEEIs on the same cake $C$. Then for all $Z\subseteq C$: $Q(Z) = P(Z)$.
\end{corollary}
\begin{proof}
	Apply Lemma \ref{ceei-decreasing-price} twice with $E=\emptyset$: once to prove that $Q(Z)\leq P(Z)$ and another time to prove that $Q(Z)\geq P(Z)$ for all $Z\subseteq C$.
\end{proof}

Define the \emph{SCEEI division rule} as the rule that selects all allocations $X$ for which $(X,P_X)$ is a SCEEI.
\begin{corollary}
	\label{ceei-single-valued}
	The SCEEI division rule is essentially-single-valued.
\end{corollary}
\begin{proof}
	By Corollary \ref{ceei-same-price}, every SCEEI on the same cake has the same price-measure $P$. Hence, every SCEEI has the same budget set (the subsets $Z\subseteq C$ for which $P(Z)=1$ are the same in every SCEEI). In every SCE, all agents attain the maximum utility in their budget sets, which is the same in all SCEEIs.
\end{proof}
\begin{remark}
	\cite{Gale1976Linear} proved that, in any exchange economy with a finite number of goods and linear utilities, the SCEEI rule is essentially-single-valued.
	Gale's result can be seen as a special case of our Corollary \ref{ceei-single-valued}, since Gale's economy is equivalent to cake-cutting when the cake is piecewise-homogeneous (each homogeneous region in such a cake represents a commodity in Gale's economy).
\end{remark}

\begin{corollary}
	\label{ceei-RM}
	The SCEEI division rule is resource-monotonic.
\end{corollary}
\begin{proof}
Suppose a cake $C$ is enlarged, let $C\cup E$ be the enlarged cake and let $(X,P)$, $(Y,Q)$ be the equilibria on $C$ and $C\cup E$, respectively. By Lemma \ref{ceei-decreasing-price}, the prices on the original cake $C$ in the new price system ($Q$) are weakly lower than in the old price system ($P$). Hence, the budget set under $Q$ contains the budget set under $P$. Hence, in the new equilibrium $Y$, all agents can afford the pieces that they had in equilibrium $X$.  In every SCE, all agents attain the maximum utility in their budget sets, which is weakly larger in $Y$ than in $X$.
\end{proof}

\begin{corollary}
	\label{ceei-PM}
	The SCEEI division rule is population-monotonic.
\end{corollary}
\begin{proof}
	Let$(X,P)$ be a SCEEI allocation. Suppose an agent $i$ leaves and abandons his share $X_i$. The pair $(X,P)$ still satisfies the SCE and EI conditions on the cake $C\setminus X_i$, so it is still a SCEEI allocation. Now, apply the resource-monotonicity (Corollary \ref{ceei-RM}) with $E=X_i$.
\end{proof}

\begin{remark}
	The CEEI rule was one of the first examples of a division rule which is \emph{not} resource-monotonic \citep{Aumann1974Note} and not population-monotonic \citep{Chichilnisky1987Walrasian}. However, as noted later by \cite{Moulin1988}, the examples use complementary products, while the standard cake-cutting model (without connectivity) assumes no complementarities between different parts of the cake.
\end{remark}

\subsection{SCEEI  and Nash-optimal are the same rule}
We saw two rules that are ESV, RM, PM, PO and PROP: the Nash-optimal rule and the SCEEI rule. In this section we prove that they are equivalent --- they always select the same allocations.
\begin{theorem}
\label{thm:nash-sceei}
(a) If $X$ is a Nash-optimal allocation, then there exists a price-measure $P_X$ such that  $(X,P_X)$ is a SCEEI.

(b) If $(X,P)$ is a SCEEI, then the allocation $X$ is Nash-optimal.
\end{theorem}

\begin{proof}[Proof of part (a)] Given a Nash-optimal allocation $X$, we define the price-measure $P_X$ in the following way:\footnote{
\citet{Weller1985Fair} uses a similar price-measure in his proof, but does not consider Nash-optimal allocations.
}
\begin{align*}
\forall i:~~
\forall Z_i \subseteq X_i:~~
P_X(Z_i) = \frac{v_i(Z_i)}{v_i(X_i)}
\end{align*}
The value of $P_X$ for slices that intersect more than one $X_i$ is uniquely determined by additivity, i.e:
\begin{align*}
\forall Z\subseteq C:~~
P_X(Z)
= \sum_{i\in N}P_X(Z\cap X_i)
= \sum_{i\in N}\frac{v_i(Z\cap X_i)}{v_i(X_i)}
\end{align*}
Since proportional allocations exist, the maximum Nash welfare is necessarily positive. Hence, all agents have a strictly positive value in $X$ ($\vrel_i(X_i) > 0$ for all $i$), so $P_X$ is well-defined. We now prove that $(X,P_X)$ satisfy each of the three conditions in the definition of SCEEI.

First, we have to prove that $P_X(Z)>0$ iff $Z$ is a positive-slice, i.e, has a positive value for at least one agent.

If $P_X(Z)>0$ then for at least one agent $i$, $P_X(Z\cap X_i) > 0$. By definition of $P_X$,
$
P_X(Z\cap X_i) = \frac{v_i(Z\cap X_i)}{v_i(X_i)}
$. Hence $v_i(Z\cap X_i)>0$
so $Z$ is indeed a positive slice.

Conversely, if $P_X(Z)=0$ then for all $i$, $P_X(Z\cap X_i)=0$. So by definition of $P_X$, $v_i(Z\cap X_i)=0$, i.e, $Z\cap X_i$ is worthless for agent $i$.
But then $Z\cap X_i$ must be worthless for \emph{all} agents ---  otherwise we could strictly increase the Nash welfare by giving $Z\cap X_i$ to an agent $j\neq i$ who values it positively.
The same is true for every $i$. Therefore, $Z$ is worth 0 for all agents, so it is indeed not a positive slice.

The EI condition is satisfied since by definition $P_X(X_i)=\vrel_i(X_i)/\vrel_i(X_i)=1$.

It remains to prove the SCE condition.
We fix an agent $i$, a positive slice $Z_i\subseteq X_i$, and a positive slice $Z\subset C$.
We show that $\vrel_i(Z_i)/P_X(Z_i) \geq \vrel_i(Z)/P_X(Z)$.

The slice $Z$ can be written as a disjoint union of $n$ slices: $Z = \cup_{j=1}^{n}{Z_j}$, such that for every $j$, $Z_j=Z\cap X_j$. By Lemma \ref{lem:infi}, for every $j$:
\begin{align*}
&&
w'(\vrel_j(X_j))\cdot \vrel_j(Z_j)
&\geq
w'(\vrel_i(X_i))\cdot \vrel_i(Z_j)
\end{align*}
where the function $w$ is a logarithmic function, e.g.\ $w(x) = \ln(x)$. Therefore, $w'(x) = 1/x$ and we get:
\begin{align*}
{\vrel_j(Z_j) \over \vrel_j(X_j)}
&\geq
{\vrel_i(Z_j) \over \vrel_i(X_i)}
\end{align*}
By definition of the standard price-measure, the left-hand side equals $P_X(Z_j)$. Hence:
\begin{align*}
\vrel_i(X_i)\cdot P_X(Z_j)
&\geq
\vrel_i(Z_j)
\end{align*}
Summing the latter inequality for $j=1,\dots,n$ gives:
\begin{gather*}
\vrel_i(X_i)\cdot \sum_{j=1}^n P_X(Z_j)\geq\sum_{j=1}^n \vrel_i(Z_j)\implies \vrel_i(X_i)\cdot P_X\Big(\bigcup_{j=1}^n Z_j\Big)\geq\vrel_i\Big(\bigcup_{j=1}^n Z_j\Big)
\implies\\
 \vrel_i(X_i)\cdot P_X(Z) \geq \vrel_i(Z)
\implies\vrel_i(X_i) \geq \vrel_i(Z)/P_X(Z)
\end{gather*}
By definition of $P_X$, for every $Z_i\subseteq X_i$:
\begin{align*}
\vrel_i(Z_i) / P_X(Z_i) = \vrel_i(X_i) \geq \vrel_i(Z)/P_X(Z)
\end{align*}
so the SCE condition holds.
\end{proof}

\begin{proof}[Proof of part (b)]
Let $(X,P)$ be a SCEEI and let $Y$ be a Nash-optimal allocation. By part (a), the pair $(Y,P_Y)$ is SCEEI. By Corollary \ref{ceei-single-valued}, the SCEEI rule is essentially-single-valued. Hence, all agents have in $X$ exactly the same values that they have in $Y$. Therefore, $X$ is Nash-optimal too.
\end{proof}
\noindent
As a corollary, we learn that  a SCEEI allocation always exists. Moreover, we get an alternative proof to the theorem of \citet{Weller1985Fair}.
\footnote{
Alternative proofs can also be found in \citet{Reijnierse1998Finding} and \citet{Barbanel2005Geometry}.
}
\begin{theorem}
\label{cor:peef}
A Pareto-optimal envy-free cake-allocation always exists.
\end{theorem}
\begin{proof}

By Lemma \ref{lem:existence}, there exists a Nash-optimal allocation, $X$.

Since $X$ maximizes an increasing welfare function, it is Pareto-optimal.

By Theorem \ref{thm:nash-sceei}(a), $(X,P_X)$ is a SCEEI.

By Lemma \ref{lem:sceei-is-wceei}, $(X,P_X)$ is also a WCEEI.

By Lemma \ref{lem:wceei-ef}, $X$ is envy-free.
\end{proof}

To conclude: we saw three different division rules based on competitive equilibrium. From weak to strong they are: WCEEI, PCEEI and SCEEI. All three are envy-free and weakly-Pareto-optimal. The latter two are Pareto-optimal. Only the latter one is ESV and Nash-optimal.

\subsection{CEEI with homogeneous commodities}
At this point we would like to clarify a potential confusion regarding the relation of our results to the CEEI rule for homogeneous divisible commodities.
Suppose there is a finite set $A$ of commodities. The bundle of an agent $i$ is represented by a vector $X_i\in \mathbb{R}^A_+$, where for every commodity $a\in A$, $X_{i,a}$ denotes the amount of commodity $a$ in the bundle.
The common definition of CE in this domain is (see e.g. \citet{mas1992equilibrium}, \citet{Bogomolnaia2017}), for all $i\in N$ and $Z\in \mathbb{R}^A_+$:
\begin{itemize}
\item $P\cdot Z \le P\cdot X_i$ implies $\vrel_i(Z) \leq \vrel_i(X_i)$,
\end{itemize}
where $P$ is the price-vector.

At first glance, the CE condition looks similar to the WCE condition (Definition \ref{def:ceei}). However, it is well known that with homogeneous commodities CEEI is equivalent to the Nash-optimal rule (e.g. \citet{eisenberg1959consensus} and \citet{vazirani2007combinatorial})
while in our case, equivalence with the Nash-optimal rule only obtains for SCEEI. In fact, in Example \ref{exm:pceei-not-esv}, the allocation giving the middle slice to Alice is PCEEI (hence also WCEEI) but not Nash-optimal!

This apparent contradiction is resolved by noting a subtle difference between the definitions of CE and WCE.\footnote{
We are grateful to Fedor Sandomirskiy and an anonymous referee for their help in clarifying this issue.}
The CE condition applies to all non-negative bundles $Z\in\mathbb{R}^A_+$.
In contrast, the WCE condition applies only to slices $Z$ that are contained in the cake $C$; this is equivalent to
requiring that the bundle $Z$ (the demand of an agent) is restricted to be below the social endowment, i.e., cannot exceed the amount of commodities in the economy.
Thus CE is stronger than WCE. In fact, here is a quick proof that CE is equivalent to SCE:

\emph{CE $\implies$ SCE:}
Suppose by contradiction that SCE is violated for agent $i$. This means that there is some commodity $a$ with $P_a>0$ and $X_{i,a}>0$ (agent $i$ holds a positive amount of $a$), and another commodity $b$ with $P_b>0$, such that $v_i(a)/P_a < v_i(b)/P_b$.
Then we can create an alternative bundle $Z$ (which may be larger than the endowment) by replacing all units of $a$ with $X_{i,a}\cdot P_a / P_b$ units of $b$.
Then $Z$ costs the same as $X_i$ but is worth more for agent $i$, contradicting the CE condition.

\emph{SCE $\implies$ CE:}
Suppose by contradiction that CE is violated for agent $i$.
This means that there is a bundle $Z$ with $P\cdot Z \le P\cdot X_i$ and $\vrel_i(Z) > \vrel_i(X_i)$. But then $v_i(Z)/(P\cdot Z) > v_i(X_i)/(P\cdot X_i)$, contradicting the SCE condition. \qed

\section{The Leximin-optimal rules}
\label{sec:leximin}
The leximin-optimal division rules
are based on the principle of equality of welfare. Intuitively these rules work as follows. First we narrow down the solution set to divisions that maximize the utility of the poorest agent. Then among those solutions that satisfy this criterion, we select those which maximize the utility of the second poorest agent. We repeat this process with the third, fourth, etc.\ poorest agent.

To formally define the leximin rule we first define lexicographical ordering of real vectors.
We say that vector $y \in \mathbb{R}^n$  is \emph{lexicographically greater} than $x \in \mathbb{R}^n$ (denoted by $x \prec y$) if $x\not=y$ and there exists a number $1 \le j \le n$ such that $x_i=y_i$ if $i < j$ and $x_{j}<y_{j}.$

\begin{definition}
	(a) For a cake division $X$, define the \emph{relative-leximin-welfare vector} as a vector of length $n$ which contains the relative values of the agents under division $X$ in a non-decreasing order.
	
	(b) A cake division $Y$ is said to be \emph{relative-leximin-better} than $X$ if the relative-leximin-welfare vector of $Y$ is lexicographically greater than the relative-leximin-welfare vector of $X$.
	
	(c) A cake division $X$ is called \emph{relative-leximin-optimal} if no other division is relative-leximin-better than $X$.
	
	(d) The \emph{relative-leximin division rule} is the rule that returns all relative-leximin-optimal divisions of the cake.
\end{definition}

The terms \emph{absolute-leximin-welfare vector}, \emph{absolute-leximin-better}, \emph{absolute-leximin-optimal} and the \emph{absolute-leximin division rule} are defined analogously.

\begin{example}
	\label{exm:leximin}
	In the following cake, the absolute-leximin division rule splits the leftmost slice between Alice and Bob, giving each of them 6. The rightmost slice is given to Carl. Thus, the absolute-leximin-optimal vector is (6,6,9).
	
	\begin{table}[h!]
		\centering
		\begin{tabular}{|l|c|c|c|c|c|}
			\hline  $\vabs_A$ & 12  & 0  \\
			\hline  $\vabs_B$ & 12  & 0  \\
			\hline  $\vabs_C$ & 21 & 9  \\
			\hline
		\end{tabular}
	\end{table}
	The relative value vector of the above division is $(6/12, 6/12, 9/30)$. The corresponding relative-leximin welfare vector is $(9/30, 6/12, 6/12) = (3/10, 1/2, 1/2)$, which is not optimal. The relative-leximin rule divides the leftmost slice between all three agents, giving $1/6$ to Carl (absolute value 3.5) and $5/12$ to Alice and Bob (absolute value 5). The rightmost slice is given to Carl. The relative-leximin-optimal vector is $(5/12, 5/12, 5/12)$.
\end{example}

Lexicographic optimization as a solution concept is used in various kinds of fair division problems.
It was advocated by the famous work of \citet{rawls1971theory}. Similar ideas appear in the bargaining solutions of \citet{Kalai1975Other} and \citet{kalai1977proportional}.
Moreover, it is a limit of an infinite sequence of welfare-maximizing rules \citep[page 104]{Moulin2004Fair}.
In division of homogeneous goods, this rule is known to be population-monotonic (e.g.\ \cite[chapter 7]{Moulin2004Fair}).
	
\cite{Dubins_1961} were the first to study the leximin rule in a cake-cutting setting.
They prove that leximin-optimal cake divisions always exist.
The proof is based on the compactness of the space of value-matrices (which we denoted by $\mathbb{M}$ in Lemma \ref{lem:existence}).
The proof is equally valid for absolute and relative leximin-optimal divisions. Hence, both the absolute- and the  relative- leximin division rules are well-defined.
Moreover, these rules are both Pareto-optimal since by definition, if a division $Y$ Pareto-dominates a division $X$, then both the absolute and the relative leximin-welfare vectors of $Y$ are larger than those of $X$.

\citet{DallAglio2001DubinsSpanier,DallAglio2003Maximin,DallAglio2014Finding} presented various properties and approximation algorithms for finding leximin divisions (they call such divisions \emph{Dubins-Spanier-optimal}).
The goal of the present section is to study the monotonicity properties of the leximin-optimal rules.

In Lemma \ref{lem:envy} we proved that, in a hyper-concave welfare-maximizing allocation, a poorer agent never envies a richer agent. In a leximin allocation, a stronger lemma is true:
\begin{lemma}
	\label{lemma:richer-holds-worthless}
	For every absolute/relative-leximin-optimal division $X$, if agent $j$
	is absolutely/relatively poorer than agent $i$ then agent $j$ believes that the piece of agent $i$ is worthless: $\vabs_j(X_i)=\vrel_j(X_i)=0$.
\end{lemma}
\begin{proof}
	If this were not the case, then we could take a small bit of $X_i$ and give it to agent $j$, thus achieving an absolute/relative-leximin-better division. But this contradicts the leximin-optimality of $X$.
\end{proof}
\begin{corollary}
\label{cor:equitable}
If all agents assign a positive value to all positive slices, then in any absolute/relative-leximin-optimal division, all agents have the same absolute/relative value.
\end{corollary}

In Example \ref{exm:leximin}, in case of the absolute-leximin-optimal division, Carl is absolutely-richer than Alice and Bob, and indeed his share is worthless for both of them.

Suppose $X$ is an old division and $Y$ is a new division of the same cake. We say that an agent $i$ \emph{conceded a slice} to agent $j$ if there is a positive slice that belonged to agent $i$ in $X$ and belongs to agent $j$ in $Y$ (in other words, $X_i\cap Y_j$ has a positive value to at least one agent in $N$). If $X$ and $Y$ are both absolute/relative-leximin-optimal, then by Pareto-optimality, $X_i\cap Y_j$ has positive value to both the agent who concedes the slice ($i$) and the recipient ($j$). Hence, we have the following corollary of Lemma \ref{lemma:richer-holds-worthless}, which is true both for relative- and absolute-values:

\begin{corollary}
	\label{cor:donors}
	Let $X$ and $Y$ be two leximin-optimal divisions. If, when switching from $X$ to $Y$, agent $i$ conceded a slice to agent $j$, then in division $X$, agent $i$ is weakly-poorer than $j$, and in division $Y$, agent $i$ is weakly-richer than $j$.
\end{corollary}

\begin{lemma}
	\label{lemma:leximin-essentially-sv}
	The absolute and the relative leximin division rules are essentially single-valued.
\end{lemma}
\begin{proof}
	The proof is the same for the absolute and the relative leximin rules, so the adjectives are omitted.
	
	Let $X$ and $Y$ be two different leximin-optimal divisions. We will prove that, when switching from $X$ to $Y$, there are no "lucky" agents (agents who gain value) nor "unlucky" agents (agents who lose value).
	
	Suppose by contradiction that agent $i$ is unlucky. Then, he must have conceded a slice to at least one other agent, say $j$. By Corollary \ref{cor:donors}, $i$ is weakly-poorer than $j$ in $X$ and weakly-richer in $Y$. But this means that $j$ is also unlucky. So, all slices conceded by unlucky agents are held by other unlucky agents. Suppose the unlucky agents take back all the slices that they conceded. This has no effect on the lucky agents, but strictly increases the value of the unlucky agents, since they now have at least the value that they had in $X$. But this contradicts the optimality of $Y$. Hence, there are no unlucky agents. If $Y$ had a lucky agent without having any unlucky agent that would contradict the optimality of $X$.

	Since $X$ and $Y$ were arbitrary leximin-optimal divisions it follows that all leximin-optimal divisions have the same value vector.
\end{proof}

\begin{lemma}
	\label{lemma:leximin-is-pm}
	The absolute- and relative-leximin division rules are population-monotonic.
\end{lemma}
\begin{proof}
	Again the proof is the same for the absolute and the relative leximin rules. Thanks to Lemma \ref{lemma:leximin-essentially-sv} it is sufficient to prove downwards-PM.
	
	Let $X$ be a leximin-optimal division of $C$. Suppose that agent $i$ leaves and abandons his share $X_i$. So $X$ is now a division of $C\setminus X_i$ among the agents $N\setminus \{i\}$. If  $X_i$ is divided arbitrarily among $N\setminus \{i\}$, the result is $X^+$, a division of $C$ which is weakly leximin-better than $X$. Let $Y$ be a leximin-optimal division of $C$ among $N\setminus \{i\}$. $Y$ is weakly leximin-better than $X^+$ and hence weakly leximin-better than $X$.
	
	We are now going to prove that there are no unlucky agents in $Y$. Suppose by contradiction that agent $i$ is unlucky. Then he must have conceded a slice to an agent $j$, who must also be unlucky (as explained in Lemma \ref{lemma:leximin-essentially-sv}). All slices conceded by unlucky agents, are held by other unlucky agents. If those took back all the slices that they conceded, then all of them would be strictly better off while the lucky ones would remain unaffected. This contradicts the optimality of $Y$. Hence there are no unlucky agents and PM is proved.
\end{proof}

All our lemmata so far were true for both absolute-leximin and relative-leximin rules. The following subsections show some differences between these two rules.

\subsection{Absolute-leximin: PM and RM but not PROP nor EF}  \label{sub:absolute-leximin}

\begin{theorem}
	\label{thm:absolute-leximin}
	The absolute-leximin division rule is population-monotonic and resource-monotonic.
\end{theorem}
\begin{proof}
PM was proved in Lemma \ref{lemma:leximin-is-pm}. The proof of RM is essentially the same: the enlargement can be treated as a piece that was acquired from an agent who left the scene\footnote{Note that this argument does not work for the relative-leximin division rule. It is possible that the relative-leximin-optimal division of the enlarged cake is lexicographically smaller than the relative-leximin-optimal division of the smaller cake, since the change in the total cake values changes the order between the agents' relative values.}.
\end{proof}

Unfortunately, the absolute-leximin rule is not PROP (hence also not EF).  For instance, in Example \ref{exm:leximin}, the absolute-leximin-optimal division gives Carl a value of 9, which is only $\frac{3}{10}$ of his total cake value.

\subsection{Relative-leximin: PM and PROP but not RM nor EF}  \label{sub:relative-leximin}

\begin{theorem}
	\label{thm:leximin}
	The relative-leximin division rule is proportional and population-monotonic.
\end{theorem}
\begin{proof}PM was proved in Lemma \ref{lemma:leximin-is-pm}.
PROP holds because proportional divisions exist. The relative-leximin-welfare of a proportional division is at least $(1/n,\dots,1/n)$, hence the optimal relative-leximin-welfare vector must be at least $(1/n,\dots,1/n)$.
\end{proof}

\begin{example}
	The cake below shows that the relative-leximin rule is not RM.
	
	\begin{table}[h!]
		\centering
		\begin{tabular}{lcccccc}
			& &   &  & &  & $\blacktriangledown$ \\
			\hline \multicolumn{1}{|c}{$\vabs_A$} &\multicolumn{1}{|c}{9}&\multicolumn{1}{|c}{9}& \multicolumn{1}{|c}{0} & \multicolumn{1}{|c}{0} & \multicolumn{1}{|c}{0} & \multicolumn{1}{|c|}{ 0}   \\
			\hline \multicolumn{1}{|c}{$\vabs_B$} &\multicolumn{1}{|c}{9}&\multicolumn{1}{|c}{9}& \multicolumn{1}{|c}{0}& \multicolumn{1}{|c}{0} & \multicolumn{1}{|c}{0} & \multicolumn{1}{|c|}{ 0}   \\
			\hline \multicolumn{1}{|c}{$\vabs_C$} &\multicolumn{1}{|c}{4}&\multicolumn{1}{|c}{4}& \multicolumn{1}{|c}{10} & \multicolumn{1}{|c}{0} & \multicolumn{1}{|c}{0} & \multicolumn{1}{|c|}{ 18}  \\
			\hline \multicolumn{1}{|c}{$\vabs_D$} &\multicolumn{1}{|c}{4}&\multicolumn{1}{|c}{4}& \multicolumn{1}{|c}{0} & \multicolumn{1}{|c}{10} & \multicolumn{1}{|c}{0} & \multicolumn{1}{|c|}{ 18}  \\
			\hline \multicolumn{1}{|c}{$\vabs_E$} &\multicolumn{1}{|c}{4}&\multicolumn{1}{|c}{4}& \multicolumn{1}{|c}{0} & \multicolumn{1}{|c}{0} & \c{10} & \multicolumn{1}{|c|}{ 18}  \\
			\hline
		\end{tabular}
	\end{table}
The largest value that can be given to both Alice and Bob is 9. Hence, in the smaller cake, the optimal relative-leximin-welfare vector is $$(9/18, 9/18, 10/18, 10/18, 10/18).$$ It is attained by halving the two leftmost slices between Alice and Bob, and giving the three slices at their right to Carl, David and Eve, in that order.
	
	In the larger cake, the largest value that can be given to Carl, David and Eve from the additional slice at the right is 6. So the largest value that can be given to them from the four rightmost slices is 16. However, the total cake value has doubled for them. Hence, if Alice and Bob keep their share of 9, the relative-leximin-welfare vector changes to  $(16/36, 16/36, 16/36, 9/18, 9/18)$.
	
	Note that Alice and Bob are now relatively-richer than Carl, David and Eve, and their pieces have a positive value for them. Thus, by Lemma \ref{lemma:richer-holds-worthless}, the division in which Alice and Bob keep their current shares cannot be relative-leximin-optimal.
	\qed
\end{example}

\begin{remark}
\citep{Bogomolnaia2016Competitive}
study the CEEI and leximin rules for a finite number of divisible goods. For this case, they provide alternative proofs of non-monotonicity of relative-leximin and of resource-monotonicity of CEEI.
\end{remark}

\begin{example}
The cake below shows that the relative-leximin rule, while PROP, is not EF.

	\begin{center}
		\begin{tabular}{|l|c|c|c|c|}
			\hline  $\vabs_A$ & \cellcolor{myGreen!25}2/5  & 3/5  & 0 \\
			\hline  $\vabs_B$ & 0  & \cellcolor{myBlue!25}1/3  &  {2/3} \\
			\hline  $\vabs_{C_1}$ & 0 & 0  & \cellcolor{myOrchid!25}1   \\
			\hline  $\vabs_{C_2}$ & 0 & 0  & \cellcolor{myOrchid!25}1 \\
			\hline  $\vabs_{C_3}$ & 0 & 0  & \cellcolor{myOrchid!25}1 \\
			\hline
		\end{tabular}
	\end{center}
A $C$-agent can only get a share more than 1/3 if another $C$-agent gets less. Similarly we cannot increase the share of Bob above $1/3$ without harming the $C$-agents, and cannot increase the share of Alice without harming Bob. Thus the leximin-optimal vector is

\begin{equation*}
\Big(\frac{1}{3},\frac{1}{3},\frac{1}{3},\frac{1}{3},\frac{2}{5}\Big).
\end{equation*}
Alice envies Bob's piece, therefore the leximin rules are not envy free.\qed
\end{example}
\begin{remark}
\citet[Theorem 3.10]{DallAglio2003Maximin}
give a more complicated proof of a stronger claim: they prove that the relative-leximin rule is not EF even when all agents value all slices positively. This contradicts a claim that \citet{Weller1985Fair} attributes to \citet{Dubins_1961}. The confusion probably arose due to an ambiguity in the term ``equitable''.
\citet{Weller1985Fair} uses this term to mean ``envy-free'', while \citet{Dubins_1961} use this term to mean ``all agents have the same relative value''. Indeed, Corollary \ref{cor:equitable} shows that a relative-leximin allocation is equitable in the latter sense.
\end{remark}

\section{Conclusion and Future Work} \label{sec:conclusion}

We studied monotonicity properties in combination with classic fairness axioms of envy-freeness, proportionality and Pareto-optimality. We proved that the Nash-optimal rule is a SCEEI. Furthermore we showed that the Nash-optimal rule is the only essentially single valued rule in a large family of welfare maximizing rules that is resource monotonic and proportional. Table \ref{tab:summary} summarizes the properties of the various division rules featured in this paper.
\def\y{\c{\yy}}  
\def\n{\c{\nn}}   
\def\N{\c{\NN}}  
\def\Y{\c{\YY}}  

\def\yy{\textcolor[RGB]{0,200,0}{Yes}}  
\def\nn{\textcolor[RGB]{200,0,0}{No}}   
\def\NN{\textcolor[RGB]{0,0,200}{No*}}  
\def\YY{\textcolor[RGB]{179,109,86}{Y.c.u.}}  
\def\upw{\textcolor[RGB]{190,125,219}{Upw}} 
\def\c#1{\multicolumn{1}{|c|}{#1}}  
\def\h#1{\c{\footnotesize #1}}  
\begin{table}[h!]
\footnotesize\centering
\begin{tabular}{lccccccc}
\hline
\c{\textbf{Rule}} & \h{ESV} &\h{EF} &\h{PROP} &\h{PO}  &\h{RM} &\h{PM}  \\ \hline
\c{absolute-leximin}    & \y   &\n & \n & \y &  \y & \y  \\ \hline
\c{relative-leximin}    & \y  & \n & \y & \y &  \n & \y  \\ \hline
\c{absolute-utilitarian}    & \n   &\n & \n & \y &  \y & \y  \\ \hline
\c{relative-utilitarian}    & \n  & \n & \n & \y &  \n & \y  \\ \hline
\c{Nash-optimal/SCEEI}    & \y  & \y & \y & \y &  \y & \y  \\ \hline
\end{tabular}
\protect\caption{\label{tab:summary}Properties of division rules presented in this paper. 
}

\end{table}

The present paper opens up many interesting research questions.

\begin{itemize}
\item We proved the uniqueness of the Nash-optimal rule in the family of welfare-maximizers. We believe that this family, in itself, can be characterized by adding axioms used in the cardinal-welfarism framework, such as anonymity and separability (see chapter 3 of \cite{Moulin2004Fair}). We have not done so since we wanted to keep the present paper focused on the fairness axioms, but this is an interesting direction for future work.
\item When dividing resources such as time or land, it may be important that the pieces are connected. In an accompanying technical, we show that when the agents insist on receiving a connected piece, no proportional and Pareto-optimal rule can be either resource-monotonic or population-monotonic report \citep{ourArxivPaperConnected}. If Pareto-optimality is relaxed to weak-Pareto-optimality, then there exists a resource-monotonic division procedure for two agents and a population-monotonic division rule for $n$ agents. It is an open question whether there exist proportional resource-monotonic rules for three or more agents with the connectivity constraint.
\item Monotonicity properties may prevent some but not all possibilities of strategic manipulation. In this paper we ignored strategic considerations and assumed that all agents truthfully report their valuations.
Indeed, \citet{Branzei2015Dictatorship} show that, under mild technical conditions, any deterministic truthful cake-cutting mechanism leaves one agent with no cake at all, so truthfulness and fairness can be combined only by a randomized mechanism.
An interesting future research topic is how to ensure monotonicity in such mechanisms. Recently, \cite{Bogomolnaia2016Competitive} discussed certain strategic properties of CEEI in the context of homogeneous divisible goods\footnote{We are grateful to an anonymous reviewer for this comment.}. It is reasonable to assume that these properties transfer to the cake-cutting setting, but this requires further work.
\item Finally, there is the question of how to compute Nash-optimal allocations. When the cake is piecewise homogeneous, the algorithm of \cite{Aziz2014Cake} finds Nash-optimal allocations in polynomial time. It is unclear how fast is their method exactly, but they use the primal-dual scheme for convex programs developed by \cite{Devanur2008}, which runs in $\mathcal{O}(n^4(\log n + n\log U + \log M))$ time (where $U$ is the maximal utility any agent can achieve, and $M$ is the total amount of money). It is a question how Nash-optimal allocations can be found when the valuations are not piecewise-constant.

\end{itemize}

\section{Acknowledgments}
This paper was born in the COST Summer School on Fair Division in Grenoble, 7/2015 (FairDiv-15). We are grateful to COST and the conference organizers for the wonderful opportunity to meet with fellow researchers from around the globe. In particular, we are grateful to Ioannis Caragiannis, Ulle Endriss and Christian Klamler for sharing their insights on cake-cutting with us. We are also thankful to Marcus Berliant, Shiri Alon-Eron, Herve Moulin, Fedor Sandomirskiy, Christian Blatter, Ilan Nehama, Peter Kristel, Alex Ravsky and Kavi Rama Murthy for their very helpful comments.

The authors acknowledge the support of the `Momentum' Programme (LP-004/2010) of the Hungarian Academy of Sciences, the \'UNKP-17-4-I, New National Excellence Program of the Ministry of Human Capacities the OTKA grant K124550, 
the ISF grants 1083/13 and 1394/16, the Doctoral Fellowships of Excellence Program, the Wolfson Chair and the Mordecai and Monique Katz Graduate Fellowship Program at Bar-Ilan University.  

This research was partially supported by Pallas Athene Domus Educationis Foundation. The views expressed are those of the authors' and do not necessarily reflect the official opinion of Pallas Athene Domus Educationis Foundation.

\appendix
\section*{Appendix. Proof of Technical Lemma}
\intervallemma*
\begin{proof}%
We first prove statement 1.

If $w$ is not strictly-concave in $[a,b]$, then $w'$ is not strictly-decreasing in $[a,b]$, so $w''$ is not strictly-negative in $[a,b]$. There are two cases:

\emph{Case \#1:} there exists a point where $w''$ is positive: $\exists x\in[a,b]: w''(x)>0$.
By continuity of $w''$, there exists an open interval around $x$ where $w''$ is positive. In this interval $w'$ is increasing and $w$ is convex.

\emph{Case \#2:} $w''$ is always weakly-negative: $\forall x\in[a,b]: w''(x) \leq 0$. So the function $w'$ is weakly-decreasing in $[a,b]$.
But, since $w'$ is not strictly-decreasing in $[a,b]$, there exist some $s,t$ with $a\leq s<t\leq b$ such that $w'(s)\leq w'(t)$. This means that, in the interval $[s,t]$, the function $w'$ must be constant.
In this interval, $w$ is linear, hence also convex.

The proof of the statement 2 is entirely analogous to Case \#1 above.
\end{proof}

\section*{References}
\small{
\bibliography{CCM}

\begin{thebibliography}{60}
\expandafter\ifx\csname natexlab\endcsname\relax\def\natexlab#1{#1}\fi
\expandafter\ifx\csname url\endcsname\relax
  \def\url#1{\texttt{#1}}\fi
\expandafter\ifx\csname urlprefix\endcsname\relax\def\urlprefix{URL }\fi
\providecommand{\eprint}[2][]{\url{#2}}
\providecommand{\bibinfo}[2]{#2}
\ifx\xfnm\relax \def\xfnm[#1]{\unskip,\space#1}\fi
\bibitem[{Arzi(2012)}]{Arzi2012Cake}
\bibinfo{author}{Arzi, O.}, \bibinfo{year}{2012}.
\newblock \bibinfo{title}{{Cake Cutting: Achieving Efficiency While Maintaining
  Fairness}}.
\newblock Master's thesis. Bar-Ilan University.
\newblock \bibinfo{note}{Under the supervision of Prof. Yonatan Aumann}.
\bibitem[{Arzi et~al.(2016)Arzi, Aumann and Dombb}]{Arzi2016Toss}
\bibinfo{author}{Arzi, O.}, \bibinfo{author}{Aumann, Y.},
  \bibinfo{author}{Dombb, Y.}, \bibinfo{year}{2016}.
\newblock \bibinfo{title}{{Toss one's cake, and eat it too: partial divisions
  can improve social welfare in cake cutting}}.
\newblock \bibinfo{journal}{Social Choice and Welfare} \bibinfo{volume}{46},
  \bibinfo{pages}{933--954}.
\bibitem[{Aumann and Peleg(1974)}]{Aumann1974Note}
\bibinfo{author}{Aumann, R.J.}, \bibinfo{author}{Peleg, B.},
  \bibinfo{year}{1974}.
\newblock \bibinfo{title}{{A note on Gale's example}}.
\newblock \bibinfo{journal}{Journal of Mathematical Economics}
  \bibinfo{volume}{1}, \bibinfo{pages}{209--211}.
\bibitem[{Aziz and Ye(2014)}]{Aziz2014Cake}
\bibinfo{author}{Aziz, H.}, \bibinfo{author}{Ye, C.}, \bibinfo{year}{2014}.
\newblock \bibinfo{title}{{Cake Cutting Algorithms for Piecewise Constant and
  Piecewise Uniform Valuations}}, in: \bibinfo{editor}{Liu, T.Y.},
  \bibinfo{editor}{Qi, Q.}, \bibinfo{editor}{Ye, Y.} (Eds.),
  \bibinfo{booktitle}{Web and Internet Economics}. \bibinfo{publisher}{Springer
  International Publishing}. volume \bibinfo{volume}{8877} of
  \textit{\bibinfo{series}{Lecture Notes in Computer Science}}, pp.
  \bibinfo{pages}{1--14}.
\bibitem[{Balinski and Young(1982)}]{Balinski1982}
\bibinfo{author}{Balinski, M.}, \bibinfo{author}{Young, H.P.},
  \bibinfo{year}{1982}.
\newblock \bibinfo{title}{{Fair Representation: Meeting the Ideal of One Man,
  One Vote}}.
\newblock \bibinfo{publisher}{Yale University Press}, \bibinfo{address}{New
  Haven}.
\bibitem[{Barbanel(2005)}]{Barbanel2005Geometry}
\bibinfo{author}{Barbanel, J.B.}, \bibinfo{year}{2005}.
\newblock \bibinfo{title}{{The Geometry of Efficient Fair Division}}.
\newblock \bibinfo{publisher}{Cambridge University Press}.
\bibitem[{Berliant et~al.(1992)Berliant, Thomson and Dunz}]{Berliant1992Fair}
\bibinfo{author}{Berliant, M.}, \bibinfo{author}{Thomson, W.},
  \bibinfo{author}{Dunz, K.}, \bibinfo{year}{1992}.
\newblock \bibinfo{title}{{On the fair division of a heterogeneous commodity}}.
\newblock \bibinfo{journal}{Journal of Mathematical Economics}
  \bibinfo{volume}{21}, \bibinfo{pages}{201--216}.
\bibitem[{Bogomolnaia and Moulin(2016)}]{Bogomolnaia2016Competitive}
\bibinfo{author}{Bogomolnaia, A.}, \bibinfo{author}{Moulin, H.},
  \bibinfo{year}{2016}.
\newblock \bibinfo{title}{{Competitive Fair Division under additive
  utilities}}.
\bibitem[{Bogomolnaia et~al.(2017)Bogomolnaia, Moulin, Sandomirskiy and
  Yanovskaya}]{Bogomolnaia2017}
\bibinfo{author}{Bogomolnaia, A.}, \bibinfo{author}{Moulin, H.},
  \bibinfo{author}{Sandomirskiy, F.}, \bibinfo{author}{Yanovskaya, E.},
  \bibinfo{year}{2017}.
\newblock \bibinfo{title}{Competitive division of a mixed manna}.
\newblock \bibinfo{journal}{Econometrica} \bibinfo{volume}{85},
  \bibinfo{pages}{1847--1871}.
\bibitem[{Braess(1968)}]{Braess1968}
\bibinfo{author}{Braess, D.}, \bibinfo{year}{1968}.
\newblock \bibinfo{title}{{\"{U}ber ein Paradoxen der Verkehrsplanung}}.
\newblock \bibinfo{journal}{Unternehmensforschung} \bibinfo{volume}{12},
  \bibinfo{pages}{258--268}.
\bibitem[{Br\^{a}nzei(2015)}]{Branzei2015Computational}
\bibinfo{author}{Br\^{a}nzei, S.}, \bibinfo{year}{2015}.
\newblock \bibinfo{title}{{Computational Fair Division}}.
\newblock Ph.D. thesis. Faculty of Science and Technology in Aarhus university.
\bibitem[{Br\^{a}nzei et~al.(2016)Br\^{a}nzei, Caragiannis, Kurokawa and
  Procaccia}]{Branzei2016Algorithmic}
\bibinfo{author}{Br\^{a}nzei, S.}, \bibinfo{author}{Caragiannis, I.},
  \bibinfo{author}{Kurokawa, D.}, \bibinfo{author}{Procaccia, A.D.},
  \bibinfo{year}{2016}.
\newblock \bibinfo{title}{{An Algorithmic Framework for Strategic Fair
  Division}}, in: \bibinfo{booktitle}{Proceedings of the Thirtieth AAAI
  Conference on Artificial Intelligence}, \bibinfo{publisher}{AAAI Press}. pp.
  \bibinfo{pages}{411--417}.
\bibitem[{Br\^{a}nzei et~al.(2017)Br\^{a}nzei, Gkatzelis and
  Mehta}]{Branzei2017Nash}
\bibinfo{author}{Br\^{a}nzei, S.}, \bibinfo{author}{Gkatzelis, V.},
  \bibinfo{author}{Mehta, R.}, \bibinfo{year}{2017}.
\newblock \bibinfo{title}{{Nash Social Welfare Approximation for Strategic
  Agents}}, in: \bibinfo{booktitle}{18th ACM conference on Economic and
  Computation (EC 17)}.
\newblock \bibinfo{note}{ArXiv preprint 1607.01569}. \eprint{1607.01569}.
\bibitem[{Br\^{a}nzei and Miltersen(2015)}]{Branzei2015Dictatorship}
\bibinfo{author}{Br\^{a}nzei, S.}, \bibinfo{author}{Miltersen, P.B.},
  \bibinfo{year}{2015}.
\newblock \bibinfo{title}{{A Dictatorship Theorem for Cake Cutting}}, in:
  \bibinfo{booktitle}{Proceedings of the 24th International Conference on
  Artificial Intelligence}, \bibinfo{publisher}{AAAI Press}. pp.
  \bibinfo{pages}{482--488}.
\bibitem[{Calleja et~al.(2012)Calleja, Rafels and Tijs}]{Calleja2012}
\bibinfo{author}{Calleja, P.}, \bibinfo{author}{Rafels, C.},
  \bibinfo{author}{Tijs, S.}, \bibinfo{year}{2012}.
\newblock \bibinfo{title}{Aggregate monotonic stable single-valued solutions
  for cooperative games}.
\newblock \bibinfo{journal}{International Journal of Game Theory}
  \bibinfo{volume}{41}, \bibinfo{pages}{899--913}.
\bibitem[{Caragiannis et~al.(2016)Caragiannis, Kurokawa, Moulin, Procaccia,
  Shah and Wang}]{Caragiannis2016Unreasonable}
\bibinfo{author}{Caragiannis, I.}, \bibinfo{author}{Kurokawa, D.},
  \bibinfo{author}{Moulin, H.}, \bibinfo{author}{Procaccia, A.D.},
  \bibinfo{author}{Shah, N.}, \bibinfo{author}{Wang, J.}, \bibinfo{year}{2016}.
\newblock \bibinfo{title}{{The Unreasonable Fairness of Maximum Nash Welfare}},
  in: \bibinfo{booktitle}{17th ACM conference on Economic and Computation (EC
  16)}.
\bibitem[{Caragiannis et~al.(2011)Caragiannis, Lai and
  Procaccia}]{Caragiannis2011}
\bibinfo{author}{Caragiannis, I.}, \bibinfo{author}{Lai, J.K.},
  \bibinfo{author}{Procaccia, A.D.}, \bibinfo{year}{2011}.
\newblock \bibinfo{title}{Towards more expressive cake cutting}, in:
  \bibinfo{booktitle}{Proceedings of the Twenty-Second International Joint
  Conference on Artificial Intelligence - Volume Volume One},
  \bibinfo{publisher}{AAAI Press}. pp. \bibinfo{pages}{127--132}.
\bibitem[{Chambers(2005)}]{Chambers_2005}
\bibinfo{author}{Chambers, C.P.}, \bibinfo{year}{2005}.
\newblock \bibinfo{title}{{Allocation rules for land division}}.
\newblock \bibinfo{journal}{Journal of Economic Theory} \bibinfo{volume}{121},
  \bibinfo{pages}{236--258}.
\bibitem[{Chichilnisky and Thomson(1987)}]{Chichilnisky1987Walrasian}
\bibinfo{author}{Chichilnisky, G.}, \bibinfo{author}{Thomson, W.},
  \bibinfo{year}{1987}.
\newblock \bibinfo{title}{{The walrasian mechanism from equal division is not
  monotonic with respect to variations in the number of consumers}}.
\newblock \bibinfo{journal}{Journal of Public Economics} \bibinfo{volume}{32},
  \bibinfo{pages}{119--124}.
\bibitem[{Cole et~al.(2013)Cole, Gkatzelis and Goel}]{Cole2013Mechanism}
\bibinfo{author}{Cole, R.}, \bibinfo{author}{Gkatzelis, V.},
  \bibinfo{author}{Goel, G.}, \bibinfo{year}{2013}.
\newblock \bibinfo{title}{{Mechanism Design for Fair Division: Allocating
  Divisible Items Without Payments}}, in: \bibinfo{booktitle}{Proceedings of
  the Fourteenth ACM Conference on Electronic Commerce},
  \bibinfo{publisher}{ACM}, \bibinfo{address}{New York, NY, USA}. pp.
  \bibinfo{pages}{251--268}.
\newblock \eprint{1212.1522}.
\bibitem[{Conitzer et~al.(2017)Conitzer, Freeman and Shah}]{Conitzer2017}
\bibinfo{author}{Conitzer, V.}, \bibinfo{author}{Freeman, R.},
  \bibinfo{author}{Shah, N.}, \bibinfo{year}{2017}.
\newblock \bibinfo{title}{Fair public decision making}, in:
  \bibinfo{booktitle}{Proceedings of the 2017 ACM Conference on Economics and
  Computation}, \bibinfo{publisher}{ACM}, \bibinfo{address}{New York, NY, USA}.
  pp. \bibinfo{pages}{629--646}.
\bibitem[{Dall'Aglio(2001)}]{DallAglio2001DubinsSpanier}
\bibinfo{author}{Dall'Aglio, M.}, \bibinfo{year}{2001}.
\newblock \bibinfo{title}{{The Dubins-Spanier optimization problem in fair
  division theory}}.
\newblock \bibinfo{journal}{Journal of Computational and Applied Mathematics}
  \bibinfo{volume}{130}, \bibinfo{pages}{17--40}.
\bibitem[{Dall'Aglio and Di~Luca(2014)}]{DallAglio2014Finding}
\bibinfo{author}{Dall'Aglio, M.}, \bibinfo{author}{Di~Luca, C.},
  \bibinfo{year}{2014}.
\newblock \bibinfo{title}{{Finding maxmin allocations in cooperative and
  competitive fair division}}.
\newblock \bibinfo{journal}{Annals of Operations Research}
  \bibinfo{volume}{223}, \bibinfo{pages}{121--136}.
\bibitem[{Dall'Aglio and Hill(2003)}]{DallAglio2003Maximin}
\bibinfo{author}{Dall'Aglio, M.}, \bibinfo{author}{Hill, T.P.},
  \bibinfo{year}{2003}.
\newblock \bibinfo{title}{{Maximin share and minimax envy in fair-division
  problems}}.
\newblock \bibinfo{journal}{Journal of Mathematical Analysis and Applications}
  \bibinfo{volume}{281}, \bibinfo{pages}{346--361}.
\bibitem[{Devanur et~al.(2008)Devanur, Papadimitriou, Saberi and
  Vazirani}]{Devanur2008}
\bibinfo{author}{Devanur, N.R.}, \bibinfo{author}{Papadimitriou, C.H.},
  \bibinfo{author}{Saberi, A.}, \bibinfo{author}{Vazirani, V.V.},
  \bibinfo{year}{2008}.
\newblock \bibinfo{title}{Market equilibrium via a primal--dual algorithm for a
  convex program}.
\newblock \bibinfo{journal}{J. ACM} \bibinfo{volume}{55},
  \bibinfo{pages}{22:1--22:18}.
\bibitem[{Dubins and Spanier(1961)}]{Dubins_1961}
\bibinfo{author}{Dubins, L.E.}, \bibinfo{author}{Spanier, E.H.},
  \bibinfo{year}{1961}.
\newblock \bibinfo{title}{{How to Cut A Cake Fairly}}.
\newblock \bibinfo{journal}{The American Mathematical Monthly}
  \bibinfo{volume}{68}.
\bibitem[{Dvoretzky et~al.(1951)Dvoretzky, Wald and
  Wolfowitz}]{dvoretzky1951relations}
\bibinfo{author}{Dvoretzky, A.}, \bibinfo{author}{Wald, A.},
  \bibinfo{author}{Wolfowitz, J.}, \bibinfo{year}{1951}.
\newblock \bibinfo{title}{Relations among certain ranges of vector measures}.
\newblock \bibinfo{journal}{Pacific Journal of Mathematics}
  \bibinfo{volume}{1}, \bibinfo{pages}{59--74}.
\bibitem[{Eisenberg and Gale(1959)}]{eisenberg1959consensus}
\bibinfo{author}{Eisenberg, E.}, \bibinfo{author}{Gale, D.},
  \bibinfo{year}{1959}.
\newblock \bibinfo{title}{Consensus of subjective probabilities: The
  pari-mutuel method}.
\newblock \bibinfo{journal}{The Annals of Mathematical Statistics}
  \bibinfo{volume}{30}, \bibinfo{pages}{165--168}.
\bibitem[{Gale(1976)}]{Gale1976Linear}
\bibinfo{author}{Gale, D.}, \bibinfo{year}{1976}.
\newblock \bibinfo{title}{{The linear exchange model}}.
\newblock \bibinfo{journal}{Journal of Mathematical Economics}
  \bibinfo{volume}{3}, \bibinfo{pages}{205--209}.
\bibitem[{Ghodsi et~al.(2011)Ghodsi, Zaharia, Hindman, Konwinski, Shenker and
  Stoica}]{Ghodsi2011Dominant}
\bibinfo{author}{Ghodsi, A.}, \bibinfo{author}{Zaharia, M.},
  \bibinfo{author}{Hindman, B.}, \bibinfo{author}{Konwinski, A.},
  \bibinfo{author}{Shenker, S.}, \bibinfo{author}{Stoica, I.},
  \bibinfo{year}{2011}.
\newblock \bibinfo{title}{{Dominant Resource Fairness: Fair Allocation of
  Multiple Resource Types}}, in: \bibinfo{booktitle}{Proceedings of the 8th
  USENIX Conference on Networked Systems Design and Implementation},
  \bibinfo{publisher}{USENIX Association}, \bibinfo{address}{Berkeley, CA,
  USA}. pp. \bibinfo{pages}{323--336}.
\bibitem[{Herreiner and Puppe(2009)}]{Herreiner2009}
\bibinfo{author}{Herreiner, D.}, \bibinfo{author}{Puppe, C.},
  \bibinfo{year}{2009}.
\newblock \bibinfo{title}{{Envy Freeness in Experimental Fair Division
  Problems}}.
\newblock \bibinfo{journal}{Theory and Decision} \bibinfo{volume}{67},
  \bibinfo{pages}{65--100}.
\bibitem[{Hill and Morrison(2010)}]{hill2010cutting}
\bibinfo{author}{Hill, T.P.}, \bibinfo{author}{Morrison, K.E.},
  \bibinfo{year}{2010}.
\newblock \bibinfo{title}{Cutting cakes carefully}.
\newblock \bibinfo{journal}{The College Mathematics Journal}
  \bibinfo{volume}{41}, \bibinfo{pages}{281--288}.
\bibitem[{Jain and Vazirani(2010)}]{jain2010eisenberg}
\bibinfo{author}{Jain, K.}, \bibinfo{author}{Vazirani, V.V.},
  \bibinfo{year}{2010}.
\newblock \bibinfo{title}{Eisenberg--gale markets: Algorithms and
  game-theoretic properties}.
\newblock \bibinfo{journal}{Games and Economic Behavior} \bibinfo{volume}{70},
  \bibinfo{pages}{84--106}.
\bibitem[{Kalai(1977)}]{kalai1977proportional}
\bibinfo{author}{Kalai, E.}, \bibinfo{year}{1977}.
\newblock \bibinfo{title}{Proportional solutions to bargaining situations:
  interpersonal utility comparisons}.
\newblock \bibinfo{journal}{Econometrica: Journal of the Econometric Society} ,
  \bibinfo{pages}{1623--1630}.
\bibitem[{Kalai and Smorodinsky(1975)}]{Kalai1975Other}
\bibinfo{author}{Kalai, E.}, \bibinfo{author}{Smorodinsky, M.},
  \bibinfo{year}{1975}.
\newblock \bibinfo{title}{Other solutions to nash's bargaining problem}.
\newblock \bibinfo{journal}{Econometrica} \bibinfo{volume}{43},
  \bibinfo{pages}{513--518}.
\bibitem[{Kash et~al.(2014)Kash, Procaccia and Shah}]{Kash2014}
\bibinfo{author}{Kash, I.}, \bibinfo{author}{Procaccia, A.D.},
  \bibinfo{author}{Shah, N.}, \bibinfo{year}{2014}.
\newblock \bibinfo{title}{Journal of artificial intelligence research}.
\newblock \bibinfo{journal}{Journal of Artificial Intelligence Research}
  \bibinfo{volume}{51}, \bibinfo{pages}{579--603}.
\bibitem[{K\'{o}czy et~al.(2017)K\'{o}czy, Bir\'{o} and Sziklai}]{Koczy2017}
\bibinfo{author}{K\'{o}czy, L.{\'{A}}.}, \bibinfo{author}{Bir\'{o}, P.},
  \bibinfo{author}{Sziklai, B.}, \bibinfo{year}{2017}.
\newblock \bibinfo{title}{{US vs.\ European apportionment practices: The
  conflict between monotonicity and proportionality}}, in:
  \bibinfo{editor}{Endriss, U.} (Ed.), \bibinfo{booktitle}{Trends in
  Computational Social Choice}. \bibinfo{publisher}{AI Access}.
  chapter~\bibinfo{chapter}{16}.
\bibitem[{Mas-Colell(1992)}]{mas1992equilibrium}
\bibinfo{author}{Mas-Colell, A.}, \bibinfo{year}{1992}.
\newblock \bibinfo{title}{Equilibrium theory with possibly satiated
  preferences}, in: \bibinfo{booktitle}{Equilibrium and Dynamics}.
  \bibinfo{publisher}{Springer}, pp. \bibinfo{pages}{201--213}.
\bibitem[{Moulin(2004)}]{Moulin2004Fair}
\bibinfo{author}{Moulin, H.}, \bibinfo{year}{2004}.
\newblock \bibinfo{title}{{Fair Division and Collective Welfare}}.
\newblock \bibinfo{publisher}{The MIT Press}.
\bibitem[{Moulin and Thomson(1988)}]{Moulin1988}
\bibinfo{author}{Moulin, H.}, \bibinfo{author}{Thomson, W.},
  \bibinfo{year}{1988}.
\newblock \bibinfo{title}{{Can everyone benefit from growth?}}
\newblock \bibinfo{journal}{Journal of Mathematical Economics}
  \bibinfo{volume}{17}, \bibinfo{pages}{339--345}.
\bibitem[{Nash(1950)}]{Nash1950Bargaining}
\bibinfo{author}{Nash, J.F.}, \bibinfo{year}{1950}.
\newblock \bibinfo{title}{{The Bargaining Problem}}.
\newblock \bibinfo{journal}{Econometrica} \bibinfo{volume}{18}.
\bibitem[{Peleg and Sudh\"olter(2007)}]{Peleg2007}
\bibinfo{author}{Peleg, B.}, \bibinfo{author}{Sudh\"olter, P.},
  \bibinfo{year}{2007}.
\newblock \bibinfo{title}{Introduction to the Theory of Cooperative Games}.
\newblock \bibinfo{publisher}{Springer}, \bibinfo{address}{Heidelberg}.
\bibitem[{Procaccia(2016)}]{Procaccia2015Cake}
\bibinfo{author}{Procaccia, A.D.}, \bibinfo{year}{2016}.
\newblock \bibinfo{title}{{Cake Cutting Algorithms}}, in:
  \bibinfo{editor}{Brandt, F.}, \bibinfo{editor}{Conitzer, V.},
  \bibinfo{editor}{Endriss, U.}, \bibinfo{editor}{Lang, J.},
  \bibinfo{editor}{Procaccia, A.D.} (Eds.), \bibinfo{booktitle}{Handbook of
  Computational Social Choice}. \bibinfo{publisher}{Cambridge University
  Press}. chapter~\bibinfo{chapter}{13}.
\bibitem[{Rawls(1971)}]{rawls1971theory}
\bibinfo{author}{Rawls, J.}, \bibinfo{year}{1971}.
\newblock \bibinfo{title}{A Theory of Justice}.
\bibitem[{Reijnierse and Potters(1998)}]{Reijnierse1998Finding}
\bibinfo{author}{Reijnierse, J.H.}, \bibinfo{author}{Potters, J.A.M.},
  \bibinfo{year}{1998}.
\newblock \bibinfo{title}{{On finding an envy-free Pareto-optimal division}}.
\newblock \bibinfo{journal}{Mathematical Programming} \bibinfo{volume}{83},
  \bibinfo{pages}{291--311}.
\bibitem[{Schilling and Stoyan(2016)}]{schilling2016continuity}
\bibinfo{author}{Schilling, R.L.}, \bibinfo{author}{Stoyan, D.},
  \bibinfo{year}{2016}.
\newblock \bibinfo{title}{Continuity assumptions in cake-cutting}.
\newblock \bibinfo{journal}{arXiv preprint arXiv:1611.04988} .
\bibitem[{Segal-Halevi(2016)}]{segal2016re}
\bibinfo{author}{Segal-Halevi, E.}, \bibinfo{year}{2016}.
\newblock \bibinfo{title}{How to re-divide a cake fairly}.
\newblock \bibinfo{journal}{arXiv preprint arXiv:1603.00286} .
\bibitem[{Segal-Halevi and Sziklai(2015)}]{ourTechReport}
\bibinfo{author}{Segal-Halevi, E.}, \bibinfo{author}{Sziklai, B.},
  \bibinfo{year}{2015}.
\newblock \bibinfo{title}{Resource-monotonicity and Population-monotonicity in
  Cake-cutting}.
\newblock \bibinfo{type}{IEHAS Discussion Papers} \bibinfo{number}{MT-DP
  2015/52}. Institute of Economics, Centre for Economic and Regional Studies,
  Hungarian Academy of Sciences.
\bibitem[{S\"{o}nmez(1994)}]{Sonmez1994}
\bibinfo{author}{S\"{o}nmez, T.}, \bibinfo{year}{1994}.
\newblock \bibinfo{title}{Consistency, monotonicity, and the uniform rule}.
\newblock \bibinfo{journal}{Economics Letters} \bibinfo{volume}{46},
  \bibinfo{pages}{229--235}.
\bibitem[{Steinhaus(1948)}]{Steinhaus1948}
\bibinfo{author}{Steinhaus, H.}, \bibinfo{year}{1948}.
\newblock \bibinfo{title}{{The problem of fair division}}.
\newblock \bibinfo{journal}{Econometrica} \bibinfo{volume}{16},
  \bibinfo{pages}{101--104}.
\bibitem[{Stromquist and Woodall(1985)}]{Stromquist1985Sets}
\bibinfo{author}{Stromquist, W.}, \bibinfo{author}{Woodall, D.R.},
  \bibinfo{year}{1985}.
\newblock \bibinfo{title}{{Sets on which several measures agree}}.
\newblock \bibinfo{journal}{Journal of Mathematical Analysis and Applications}
  \bibinfo{volume}{108}, \bibinfo{pages}{241--248}.
\bibitem[{Sziklai and Segal-Halevi(2015)}]{ourArxivPaperConnected}
\bibinfo{author}{Sziklai, B.}, \bibinfo{author}{Segal-Halevi, E.},
  \bibinfo{year}{2015}.
\newblock \bibinfo{title}{{Resource-monotonicity and Population-monotonicity in
  Connected Cake-cutting}}.
\newblock \bibinfo{note}{{arXiv} preprint 1510.05229}.
\bibitem[{Tasn\'adi(2003)}]{Tasnadi2003}
\bibinfo{author}{Tasn\'adi, A.}, \bibinfo{year}{2003}.
\newblock \bibinfo{title}{{A new proportional procedure for the n-person
  cake-cutting problem}}.
\newblock \bibinfo{journal}{Economics Bulletin} \bibinfo{volume}{4},
  \bibinfo{pages}{1--3}.
\bibitem[{Thomson(1997)}]{Thomson1997Replacement}
\bibinfo{author}{Thomson, W.}, \bibinfo{year}{1997}.
\newblock \bibinfo{title}{{The Replacement Principle in Economies with
  Single-Peaked Preferences}}.
\newblock \bibinfo{journal}{Journal of Economic Theory} \bibinfo{volume}{76},
  \bibinfo{pages}{145--168}.
\bibitem[{Thomson(2011)}]{Thomson_2011}
\bibinfo{author}{Thomson, W.}, \bibinfo{year}{2011}.
\newblock \bibinfo{title}{{Fair Allocation Rules}}, in:
  \bibinfo{booktitle}{Handbook of Social Choice and Welfare}.
  \bibinfo{publisher}{Elsevier {BV}}, pp. \bibinfo{pages}{393--506}.
\bibitem[{Thomson(2012)}]{Thomson_2012}
\bibinfo{author}{Thomson, W.}, \bibinfo{year}{2012}.
\newblock \bibinfo{title}{{On The Axiomatics Of Resource Allocation:
  Interpreting The Consistency Principle}}.
\newblock \bibinfo{journal}{Economics and Philosophy} \bibinfo{volume}{28},
  \bibinfo{pages}{385--421}.
\bibitem[{Vazirani(2007)}]{vazirani2007combinatorial}
\bibinfo{author}{Vazirani, V.V.}, \bibinfo{year}{2007}.
\newblock \bibinfo{title}{Combinatorial algorithms for market equilibria}, in:
  \bibinfo{booktitle}{Algorithmic Game Theory}, pp. \bibinfo{pages}{103--134}.
\bibitem[{Walsh(2011)}]{Walsh2011Online}
\bibinfo{author}{Walsh, T.}, \bibinfo{year}{2011}.
\newblock \bibinfo{title}{{Online Cake Cutting}}.
\newblock \bibinfo{journal}{Algorithmic Decision Theory}
  \bibinfo{volume}{6992}, \bibinfo{pages}{292--305}.
\bibitem[{Weller(1985)}]{Weller1985Fair}
\bibinfo{author}{Weller, D.}, \bibinfo{year}{1985}.
\newblock \bibinfo{title}{{Fair division of a measurable space}}.
\newblock \bibinfo{journal}{Journal of Mathematical Economics}
  \bibinfo{volume}{14}, \bibinfo{pages}{5--17}.
\bibitem[{Young(1987)}]{Young1987}
\bibinfo{author}{Young, H.}, \bibinfo{year}{1987}.
\newblock \bibinfo{title}{On dividing an amount according to individual claims
  or liabilities}.
\newblock \bibinfo{journal}{Mathematics of Operations Research}
  \bibinfo{volume}{12}, \bibinfo{pages}{65--72}.

\end{thebibliography}
}
\end{document}